\documentclass[a4paper]{article}
\usepackage{arxiv}

\usepackage{cite}
\usepackage{array}
\usepackage[table]{xcolor}
\usepackage[normalem]{ulem}
\usepackage{hyperref}
\hypersetup{
    colorlinks,
    linkcolor={black},
    citecolor={black},
    urlcolor={blue!80!black}
}
\usepackage{appendix}
\usepackage{enumerate}
\usepackage[shortlabels]{enumitem}
\usepackage[acronym,nohypertypes={acronym,notation}]{glossaries}
\newacronym{3gpp}{3GPP}{3rd Generation Partnership Project}
\newacronym{ai}{AI}{Artificial Intelligence}
\newacronym{aoii}{AoII}{Age of Incorrect Information}
\newacronym{aoi}{AoI}{Age of Information}
\newacronym{bs}{BS}{base station}
\newacronym{dt}{DT}{digital twin}
\newacronym{drx}{DRX}{discontinuous reception }
\newacronym{ppp}{PPP}{Poisson Point Process}
\newacronym{su}{SU}{scheduling unit}
\newacronym{ofdm}{OFDM}{Orthogonal Frequency Division Multiplexing}
\newacronym{wud}{WuD}{wake-up radio device}
\newacronym{wus}{WuS}{wake-up signal}
\newacronym{idwu}{IDWu}{identity-based wake-up}
\newacronym{fifo}{FIFO}{First In First Out}
\newacronym{awgn}{AWGN}{Additive White Gaussian Noise}
\newacronym{mac}{MAC}{Medium Access Control}
\newacronym{pdf}{pdf}{probability density function}
\newacronym{rb}{RB}{resource block}
\newacronym{ucwu}{UCWu}{UniCast Wake-up}
\newacronym{ul}{UL}{uplink}
\newacronym{ue}{UE}{User equipment}
\newacronym{dl}{DL}{downlink}
\newacronym{harq}{HARQ}{hybrid automatic repeat request}
\newacronym{id}{ID}{identifier}
\newacronym{iiot}{IIoT}{Industrial Internet of Things}
\newacronym{ack}{ACK}{Acknowledgement}
\newacronym{mmtc}{mMTC}{massive Machine Type Communication}
\newacronym{mse}{MSE}{mean square error}
\newacronym{urllc}{URLLC}{ultra-reliable low-latency communication}
\newacronym{pmf}{pmf}{probability mass function}
\newacronym{iot}{IoT}{Internet of Things}
\newacronym{cps}{CPS}{Cyber-Physical System}
\newacronym{nr}{NR}{New Radio}
\newacronym[\glsshortpluralkey={r.vs.}]{rv}{r.v.}{random variable}
\newacronym{ss}{SS}{Synchronization Signal}
\newacronym{wur}{WuR}{wake-up receiver}
\newacronym{pdcch}{PDCCH}{Physical Downlink Control Channel}
\newacronym{5G}{5G}{Fifth Generation}
\newacronym{voi}{VoI}{Value of Information}
\newacronym{wsn}{WSN}{wireless sensor network}
\newacronym{xr}{XR}{eXtended Reality}
\newacronym{mcs}{MCS}{Modulation and Coding Scheme}
\newacronym{fsa}{FSA}{Framed Slotted ALOHA}
\usepackage{tabularx}
\usepackage{booktabs}
\usepackage[most]{tcolorbox}
\newcolumntype{Y}{>{\centering\arraybackslash}X}
\newcolumntype{?}{!{\vrule width 1.5pt}}
\usepackage{algorithm}
\usepackage{algpseudocode}
\algnewcommand{\LeftComment}[1]{\Statex \(\triangleright\) \textit{#1}}
\usepackage{tablefootnote}
\usepackage{amsfonts,amsmath,amssymb,amsxtra,amsbsy,mathtools,cuted,bbold,amsthm}
\usepackage{bm}
\usepackage{upgreek}
\usepackage{graphicx}
\usepackage[caption=false]{subfig}
\usepackage{tikz}
\usepackage{pgfplots}
\usepackage{pgfplotstable}
\pgfplotsset{
    compat=1.3,
    legend style={font=\scriptsize, fill opacity=1,  draw opacity=1, text opacity=1, draw=white!15!black, legend cell align=left, align=left},     
    ymajorgrids=true,
    xmajorgrids=true,    
    yminorticks=false,
    xminorticks=false,
    grid style={dashed},
    title style={font=\small},
    label style={font=\footnotesize},
    tick label style={font=\footnotesize},    
    tick align=inside,
    axis background/.style={fill=white},
    ylabel shift=-5pt,
    xlabel shift=-3pt,
}
\usepgfplotslibrary{fillbetween}
\usetikzlibrary{patterns,arrows,plotmarks}
\usepgfplotslibrary{groupplots}
\pgfdeclarelayer{background}
\pgfsetlayers{background,main}
\usetikzlibrary{automata,positioning}
\usetikzlibrary{decorations}
\usetikzlibrary{shapes.arrows}
\usetikzlibrary{tikzmark}
\usetikzlibrary{calc}
\usetikzlibrary{decorations.markings}
\algrenewcommand\algorithmicindent{10pt}
\usepgfplotslibrary{colorbrewer}

\def \twowidth{0.345\columnwidth}
\def \four{0.32\columnwidth}

\usepackage{array}
\newcolumntype{C}[1]{>{\centering\arraybackslash}p{#1}}

\definecolor{amaranth}{rgb}{0.9, 0.17, 0.31}
\definecolor{steelblue}{RGB}{176,196,222}
\definecolor{darkblue}{RGB}{0,0,139}
\definecolor{lightblue}{RGB}{31,119,180}
\definecolor{deepskyblue}{RGB}{0,191,255}
\definecolor{lightskyblue}{RGB}{135,206,250}
\definecolor{lightgray}{rgb}{0.82, 0.82, 0.82}
\definecolor{gray}{RGB}{140,140,140}
\definecolor{darkgray204}{RGB}{204,204,204}
\definecolor{darkgray224}{RGB}{224,224,224}

\newcommand{\revise}[1]{\textcolor{black}{#1}}

\newtheorem{remark}{Remark}
\newtheorem{lemma}{Lemma}

\newtheorem{assumption}{Assumption}

\newcommand{\E}[1]{\mathbb{E}\left[ #1 \right]} 
\newcommand{\mc}[1]{\mathcal{#1}}   

\makeatletter
\def\@citex[#1]#2{\leavevmode
\let\@citea\@empty
\@cite{\@for\@citeb:=#2\do
{\@citea\def\@citea{,\penalty\@m\ }%
\edef\@citeb{\expandafter\@firstofone\@citeb\@empty}%
\if@filesw\immediate\write\@auxout{\string\citation{\@citeb}}\fi
\@ifundefined{b@\@citeb}{\hbox{\reset@font\bfseries ?}%
\ G@refundefinedtrue
\@latex@warning
{Citation `\@citeb' on page \thepage \space undefined}}%
{\@cite@ofmt{\csname b@\@citeb\endcsname}}}}{#1}}
\makeatother

\newtcbtheorem{Summary}{\bfseries Summary}{enhanced,drop shadow={black!50!white},
  coltitle=black,
  top=0.3in,
  attach boxed title to top left=
  {xshift=1.5em,yshift=-\tcboxedtitleheight/2},
  boxed title style={size=small,colback=lightgray}
}{summary}

\newtcolorbox[auto counter]{summary}[1][]{title={\bfseries Summary~\thetcbcounter},enhanced,drop shadow={black!50!white},
  coltitle=black,
  top=0.3in,
  attach boxed title to top left=
  {xshift=1.5em,yshift=-\tcboxedtitleheight/2},
  boxed title style={size=small,colback=white},#1}
\newtcolorbox[auto counter]{channel}[1][]{title={\bfseries Summary~\thetcbcounter},enhanced,drop shadow={black!50!white},
  coltitle=black,
  top=0.3in,
  attach boxed title to top left=
  {xshift=1.5em,yshift=-\tcboxedtitleheight/2},
  boxed title style={size=small,colback=white},#1}

\date{}

\begin{document}

\title{Dual-Mode Wireless Devices for Adaptive Pull and Push-Based Communication}

\author{
        \begin{tabular}{c}
            Sara Cavallero, 
            Fabio Saggese,
            Junya Shiraishi,
            Israel Leyva-Mayorga,  \\
            Shashi Raj Pandey,
            Chiara Buratti, 
            Petar Popovski 
        \end{tabular}      
        \thanks{S. Cavallero is  with the National Laboratory of Wireless Communications (WiLab), CNIT, Italy (e-mail: sara.cavallero@wilab.cnit.it). C. Buratti is with the Dep. of Electrical, Electronics, and Information Engineering "Guglielmo Marconi" of the University of Bologna, Italy (email: c.buratti@unibo.it).
        J. Shiraishi, I. Leyva-Mayorga, S. R. Pandey, and P. Popovski are with the Electronic Systems Dep., Aalborg University, Denmark (emails: \{srp, jush, ilm, petarp\}@es.aau.dk). F. Saggese is with Dep. of Information Engineering, University of Pisa, Italy (fabio.saggese@ing.unipi.it), and his work is funded by Horizon Europe MSCA Postdoctoral Fellowships with Grant~101204088. This work was supported, in part, by the Velux Foundation, Denmark, through the Villum Investigator Grant WATER, nr. 37793, partly by the Horizon Europe SNS ``6G-XCEL" project with Grant 101139194, and partly by the Horizon Europe SNS ``6G-GOALS'' project with grant 101139232. The work of J. Shiraishi was supported by Horizon Europe MSCA Postdoctoral Fellowships with Grant~101151067. J. Shiraishi is now with National Institute of Information and Communications Technology (NICT), Japan (jush@nict.go.jp).}}%


\maketitle

\begin{abstract}

This paper introduces a dual-mode communication framework for wireless devices that integrates query-driven (pull) and event-driven (push) transmissions within a unified time-frame structure. Devices typically respond to information requests in pull mode, but if an anomaly is detected, they preempt the regular response to report the critical condition. Additionally, push-based communication is used to proactively send critical data without waiting for a request. This adaptive approach ensures timely, context-aware, and efficient data delivery across different network conditions.
To achieve high energy efficiency, we incorporate a wake-up radio mechanism and we design a tailored medium access control (MAC) protocol that supports data traffic belonging to the different communication classes. A comprehensive system-level analysis is conducted, accounting for the wake-up control operation and evaluating three key performance metrics: the success probability of anomaly reports (push traffic), the success probability of query responses (pull traffic) and the total energy consumption. Numerical results characterize the system's behavior and highlight the inherent trade-off between push and pull success probabilities as a function of allocated communication resources. Our analysis demonstrates that the proposed approach achieves up to a 42\% reduction in energy consumption per served packet compared to traditional approaches, while maintaining reliable support for both communication paradigms.
\end{abstract}


\section{Introduction}
\label{sec:introduction}

Advances in communication technologies are enabling more adaptive and resource-efficient wireless systems. Emerging applications such as digital twins and \gls{xr} demand low-latency, energy-efficient connectivity and seamless coordination between the physical and digital domains. To meet these requirements, \revise{a critical challenge is} 
selectively collecting sensing data from devices \revise{to a central controller--e.g., a \gls{bs}--}leveraging \revise{digital models of the physical environment} to guide communication towards application-specific goals. This interaction calls for efficient two-way communication and motivates the development of new frameworks capable of balancing inherent trade-offs~\cite{OsamaEnergyLatency}. 

Within \revise{the context of data aggregation}, two fundamentally different communication modes have been extensively studied: \emph{push}-based and \emph{pull}-based communications, each presenting distinct implications for \gls{mac} protocol design~\cite{pandey2025mediumaccesspushpulldata}. Push-based communication is an \emph{event-driven approach} in which devices autonomously decide to transmit data when a predefined condition is met, such as detecting an anomaly or an environmental change. While this enables informative updates, it suffers from contention caused by uncoordinated channel access. 
Conversely, pull-based communication operates on a demand-driven basis, where a central controller requests data only when needed. This approach reduces redundant transmissions and improves resource utilization\revise{~\cite{agheli2024effective}}.
\revise{However, it introduces latency since data is retrieved only upon request, which may delay time-sensitive updates. Additionally, it relies on the central controller's \emph{belief} about the network's status, potentially overlooking sensitive changes in the environment~\cite{pandey2025mediumaccesspushpulldata}.}
\revise{A key challenge in practical deployments is that the occurrence of critical events and the identity of the devices observing them are generally unknown a priori. As a result, relying solely on pull-based communication may lead to missed or delayed detection of time-sensitive events, since data is only retrieved upon request. Conversely, purely push-based approaches can result in inefficient resource usage due to uncoordinated access and unpredictable traffic demand. This fundamental uncertainty motivates the need for communication frameworks that can simultaneously support event-driven transmissions and query-driven data collection. Nevertheless,} 
despite their respective advantages, push- and pull-based strategies are often treated separately, with limited research exploring their joint operations
~\cite{grandient-opt, agheli2025integratedpushandpullupdatemodel, talli2024push}. A hybrid approach that smartly integrates both paradigms could enable wireless systems to adapt to diverse traffic patterns, ensuring timely event-driven responses alongside efficient scheduled data retrieval.

To address this \revise{challenge from a \gls{mac} perspective}, we propose a dual-mode communication framework that integrates push and pull mechanisms within a time-slotted structure inspired by the \gls{3gpp}. Each frame is divided into two sub-frames: $i$) a pull sub-frame for scheduled retrieval \revise{of external queried data}, and $ii$) a push sub-frame where devices \revise{detecting anomalous} data contend for access using a \gls{fsa} approach~\cite{Aloha_push}. 
\revise{Building upon our earlier studies on push/pull coexistence~\cite{cavallero2024co-existence, Shiraishi_2024_globecom},  a key innovation of the proposed framework is its support \textcolor{black}{for} device-level adaptation, allowing the devices to dynamically switch between push and pull modes depending on their internal state:} 
a device that has detected an anomaly will respond to pull requests \textcolor{black}{allowing the data related to the anomaly to \emph{preempt} the content required by the query, enhancing responsiveness without experiencing} additional energy costs.
\revise{Notably, our approach enables the direct use of ultra-low-power \glspl{wur}~\cite{piyare2017ultra}, which allow continuous device availability with a minimal level of activity and thereby low energy consumption}.


We develop a model to analyze how the proposed framework balances \revise{\emph{push and pull success probabilities}} and \emph{\revise{devices'} energy consumption} under diverse traffic loads. 
\revise{Our metrics evaluate how much \emph{relevant data} can be successfully collected at the \gls{bs} under a given \emph{energy expense}, linking the chosen push/pull design to a quality of service--e.g., latency, information freshness, energy--which is important for a variety of applications.}
\revise{Numerical simulations are performed modeling pull traffic as a Poisson process with strict per-frame deadlines, while push traffic \textcolor{black}{is modeled} via both correlated and uncorrelated Binomial \textcolor{black}{processes}.\footnote{\revise{These models are widely used in \gls{iot} research to capture uncoordinated, low-rate query arrivals and sporadic events, providing analytical tractability while reasonably approximating typical device behavior~\cite{Metzger2019Modeling, RUIZGUIROLA2025126726}.}}}
Our results demonstrate the effectiveness in supporting scalable and energy-efficient wireless systems, particularly in scenarios requiring responsiveness to both scheduled and spontaneous data flows.

\subsection{Related Works}
\label{sec:related_work}

Pull-based communication is an attractive approach to reduce unnecessary data transmissions of wireless devices \revise{ by considering the data necessity from the application layer requirements~\cite{agheli2024effective}. It \textcolor{black}{has proven}} useful and efficient for a variety of \gls{iot} monitoring applications~\cite{akar2024query, kalor2025data,Agheli2026pullquery} and informative data collection for controlling mobile actuators~\cite{shiraishi2022query}. 
\revise{In this direction, pull-based communication has also been explored for remote tracking~\cite{zakeri2025goal,vilni2025real}, where goal-oriented and correlation-aware scheduling policies are designed based on the estimation error and transmission cost, to retrieve information that is only relevant to the estimation objective.}
Furthermore, in the context of \gls{iot} networks, energy efficiency is a key requirement to prolong the network lifetime. To this end, the introduction of energy-saving \textcolor{black}{technologies} for idle listening \textcolor{black}{in} communication is \textcolor{black}{fundamental}. 
In the context of cellular networks, \gls{drx}~\cite{lin2022survey} is well investigated: \gls{ue} periodically activates its main radio \textcolor{black}{following} \textcolor{black}{a} \gls{drx} cycle to check the communication requests. However, this approach faces \textcolor{black}{challenges} in managing the trade-off between the energy consumption and latency~\cite{rostami2020wake}. \textcolor{black}{To deal with this problem, the concept of \gls{wur} access has been introduced~\cite{froytlog2019ultra}: a separate ultra-low-power \gls{wur} is installed in an \gls{iot} device; the main radio is activated only when a specific \gls{wus} is detected by the \gls{wur}. Thus, \glspl{wur} enable the \gls{bs} to retrieve data from the target devices in a demand-driven way. The authors in \cite{rostami2019wake_model} analyzed the system-level performance of \glspl{wur} in terms of energy efficiency under the latency constraint. Also, the authors in~\cite{ruiz2022energy} introduced a scheme that controls the \textcolor{black}{wake-up control parameters} based on the traffic forecasting model.
Recently, the \gls{3gpp} introduced architectures of \gls{wur} and \gls{wus} in Rel-18~\cite{hoglund20243gpp,3gpp:38-869,wagner2023low}.}

In contrast, push-based communication is essential for latency-sensitive scenarios such as industrial monitoring. 
Grant-free access protocols, such as ALOHA, are commonly employed due to their simplicity.
In this context, recent advances have focused on optimizing these protocols to improve data freshness and responsiveness. For instance, \cite{Aloha_push} conducted a detailed \gls{aoi} analysis of \gls{fsa} on large-scale \gls{iot} systems, deriving optimal settings for frame size and access probability to balance freshness and system throughput.  
The authors in~\cite{Aloha_mmtc} further developed adaptive frame sizing based on \gls{aoi} to enhance timeliness, while in ~\cite{chiariotti2024distributed}, the authors proposed an access protocol employing dynamic epistemic logic to enable distributed sensors to estimate each other's \gls{aoii} and coordinate transmissions. 

Recent studies have begun to explore the interplay between push- and pull-based communication. 
In~\cite{talli2024push}, the authors model the trade-off between push and pull communications in cyber-physical systems as a Markov decision process, providing an optimal strategy maximizing the \gls{voi}. However, their focus is on selecting the most effective global communication policy (push or pull) for a given system scenario. 
Similarly, in~\cite{agheli2025integratedpushandpullupdatemodel}, the authors present an integrated push-and-pull update model optimized for effective goal-oriented communication.  
While this work conceptually bridges the two paradigms, it focuses on system-wide optimization rather than per-device adaptation.~\cite{chiariotti2026combined} \revise{analyzes the performance of a push-pull solution for the tasks of digital twin alignment (pull) and anomaly reporting (push) with no regard \textcolor{black}{for} the energy efficiency of the system.}

\revise{To the best of our knowledge, \textcolor{black}{our earlier works~\cite{cavallero2024co-existence, Shiraishi_2024_globecom} were the only ones focusing on \gls{mac} structure enabling \gls{wur} technology.}}
In \cite{cavallero2024co-existence}, we considered a system employing \glspl{wur} with \gls{idwu}, while \cite{Shiraishi_2024_globecom} adopted a content-based wake-up mechanism, where activation depends on the relevance of the devices' data. Those studies focused on coordinating distinct devices, each operating exclusively in either push or pull mode\revise{, without 
support for device-level dynamic adaptation.} 
\textcolor{black}{The system and \gls{mac} frame designed in this paper overcome this limitation} by allowing each node to autonomously select its communication mode, i.e., push or pull, based on the importance of observed data and on real-time conditions, offering greater responsiveness and efficiency.

\subsection{Contributions}
\label{sec:contributions}
The main contribution of this paper is the design and analysis of a novel dual-mode communication framework that enables per-device adaptability between push- and pull-based communication in wireless access networks. 
Specifically, the contribution\textcolor{black}{s} of this paper can be summarized as follows:
\begin{itemize}
    \item   \emph{Adaptive dual-mode capability}: A unified communication model is proposed where each wireless device can act in push- or pull-mode depending on real-time context (e.g., presence of an anomaly or external query), enabling flexible and efficient operation.
    \item \emph{Unified \gls{mac} frame structure}: 
    \revise{A \gls{3gpp}-inspired time-slotted \gls{mac} frame structure enabling \gls{wur} implementation is introduced};
    the frame is divided into pull and push sub-frames, offering support for data transmission from different communication mode\textcolor{black}{s}.
    \item \emph{Modeling and system-level performance characterization}: 
    A stochastic analytical model is developed to evaluate system-level performance in terms of \revise{push and pull} success probabilities and \revise{devices'} energy consumption. Numerical evaluation validates the analytical model, \textcolor{black}{highlights} the trade-off between push/pull performance based on resource allocation, and determines the feasible operating region for supporting different levels of traffic. 
    
    \item \emph{Energy-efficient design}: \revise{We prove how} the integration of low-power \glspl{wur} at the devices and \glspl{wus} in the associated \gls{mac} protocol can reduce the overall energy consumption \revise{w.r.t. traditional approaches}.  
\end{itemize}

The remainder of the paper is organized as follows. Sec.~\ref{sec:system-model} describes the system model, outlining the integration of push and pull mechanisms and the role of \glspl{wur}. Sec.~\ref{sec:performance} presents a detailed mathematical analysis of the proposed system performance, focusing on success probability and energy consumption. Sec.~\ref{sec:results} presents numerical results obtained through simulations. Finally, Section \ref{sec:conclusion} concludes the paper. 

\paragraph*{Notations}
$\mc{P}(n = i)$ represents the \gls{pmf} of \gls{rv} $n$; $n|m$ denotes the \gls{rv} $n$ conditioned to the knowledge of $m$, while $\mc{P}(n=i | m=k)$ is its \gls{pmf}. 
$\E{\cdot}$ represents the expectation operator. 
$\lfloor x \rfloor$ denotes the greatest integer less than or equal to $x$.

\section{System Model}
\label{sec:system-model}
We consider a \revise{\gls{wsn}} where a \gls{bs} manages access control for $N$ devices, also referred to as nodes, communicating over a shared wireless medium. Each node has three main components: 1) an always-on sensor that continuously monitors a specific process and stores the collected data locally; 2) an ultra-low-power \gls{wur}, which may remain active or inactive, and monitors incoming communication requests for activating the main radio (cf. Sec.~\ref{sec:model:energy}); and 3) a main radio to transmit its data packets, which remains inactive until it is triggered by the \gls{wur} or by anomalies detected by the sensor -- cf. Sec.~\ref{sec:model:anomalies} and Sec.~\ref{sec:model:energy}.

The system operates in two main modes to retrieve data from the sensors, namely, \emph{pull-based} and \emph{push-based} communication modes, illustrated in Fig.~\ref{fig:system_model}. \emph{Pull-based} communication is initiated by an application in the cloud, which transmits \emph{queries} towards the \gls{bs} requesting data from specific devices in the sensor network. In turn, the \gls{bs} transmits unicast \gls{idwu} signals to activate the main radio interface of those devices; the target node \gls{id} is embedded into each \gls{wus}~\cite{3gpp:38-869}. Upon receiving a \gls{wus} that matches the node's ID, the node activates its main radio and transmits its latest observed data to the \gls{bs}. Conversely, \emph{push-based} communication is initiated by the nodes' sensor when a pre-defined condition is met, such that the node generates data that is critical for the application, and enters an \emph{alarm state}. When a node enters an alarm state, it activates its main radio interface and attempts to transmit its data in the next transmission opportunity.

In the following, we present the detailed model and assumptions about the \gls{mac} frame structure, the queries, anomalies, and the energy consumption.
The main variables used are summarized in Table~\ref{tab:notation}.

\begin{figure}[t]
    \centering
    \includegraphics[width=.7\columnwidth]{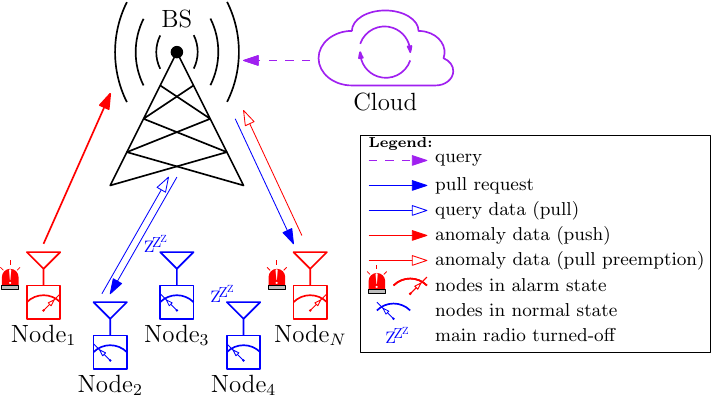}
    \caption{Toy-example of the scenario with $N=5$ nodes. Device 1 is in the alarm state and has autonomously activated to transmit (push-based mode); device 2 has received a pull request from the \gls{bs} and sends the requested data after waking up its main radio (pull-based mode); devices 3 and 4 keep their main radios off not having received pull requests or detected any anomaly; device 5 has received a pull request while in the alarm state and it sends the anomaly data instead of the requested data in the reply (pull-based mode).}
    \label{fig:system_model}
\end{figure}

\begin{table}[htb]
    \centering
    \caption{Main variables used throughout the paper.}
    \begin{tabular}{c | p{0.8\columnwidth}}
    \toprule
    \multicolumn{2}{c}{\textbf{Frame}} \\
    \midrule
    $T$ & Frame duration [s] \\
    $F$ & No. of slots in a frame \\
    $\tau$ & Slot duration [s] \\
    $\beta_t$ & Fraction of slot duration for UL data payload \\
    $\beta_r$ & Fraction of slot duration for ACK reception \\
    $k_w$ & No. of slots needed for a \gls{wus} transmission \\
    $Q$ & No. of \glspl{wus} and pull slots in the pull sub-frame \\
    $k_c$ & No. of slots for Control Signal in the push sub-frame \\
    $P$ & No. of push slots in the push sub-frame \\
    \midrule
    \multicolumn{2}{c}{\textbf{Nodes and packets}} \\    
    \midrule
    $N$ & No. of nodes in the system \\
    $n_q(t)$ & No. of queries collected at the \gls{bs} at the beginning of frame $t$ \\
    $n_n(t)$ & No. of new alarm packets generated at frame $t$ attempting access in $t+1$ \\
    $n_a(t)$ & No. of alarm packets attempting access within frame $t$ \\
    $n_w(t)$ & No. of alarmed nodes receiving the \gls{wus} at frame $t$ \\
    $n_p(t)$ & No. of alarm packets attempting access during push sub-frame at frame $t$ \\
    $n_s(t)$ & No. of alarm packets that succeed access in push sub-frame  at frame $t$ \\
    $n_f(t)$ & No. of alarm packets \textcolor{black}{failing} the access in push sub-frame access at frame $t$ \\
    \midrule
    \multicolumn{2}{c}{\textbf{Traffic}} \\    
    \midrule
    $\lambda_q$ & average queries arrival rate per node [queries/s] \\
    $\lambda_a$ & average alarm arrival rate per node [alarms/s] \\
    \midrule
    \multicolumn{2}{c}{\textbf{Power consumption}} \\
    \midrule
    $\xi_w$ & Power consumed by a node's \gls{wur}\\
    $\xi_r$ & Power consumed for reception by a node's main radio\\
    $\xi_t$ & Power consumed for transmission by a node's main radio\\
    \bottomrule
    \end{tabular}    
    \label{tab:notation}
\end{table}

\subsection{Frame structure and data transmission}
\label{sec:model:frame}

As shown in Fig.~\ref{fig:time-diagram}, the system operates in a time-slotted framework organized into frames, in which both \gls{ul} and \gls{dl} transmissions occur. Each frame has a duration denoted as $T$ and consists of $F$ slots, each lasting $\tau$ seconds, such that $\tau F = T$.  Inspired by the \gls{3gpp} 5G standards~\cite{3gpp:38-912}, $\tau$ is assumed to be a multiple of an \gls{ofdm} symbol, as in~\cite{cavallero_urllc, mini-slot}. We further assume that each slot is designed to accommodate a \gls{ul} \emph{data payload}, the corresponding \emph{\gls{ack}} \gls{dl} message, and a \emph{guard time} to switch between \gls{ul} and \gls{dl} transmissions, \revise{as described in~\cite{cavallero_urllc}}. The three parts occupy a pre-defined fraction of the slot duration $\tau$ denoted by $\beta_t$, $\beta_r$, and $1 - \beta_t - \beta_r$, respectively. Accordingly, the durations of the \gls{ul} and \gls{dl} transmissions are $\beta_t \tau$ and $\beta_r \tau$ seconds.
Similar to our preliminary work~\cite{cavallero2024co-existence}, each frame is further divided into two sub-frames: \textbf{pull sub-frame}, and \textbf{push sub-frame}, to accommodate pull and push \revise{transmissions}.

The \textbf{pull sub-frame} comprises $Q$ \glspl{wus}, each followed by a \emph{pull slot} dedicated to data transmission\revise{: a node receiving a \gls{wus} transmits its \emph{pull reply} in the subsequent pull slot.} \revise{Without loss of generality, we assume that the transmission of a \gls{wus} and the activation of the node's main radio require $k_w$ slots~\cite{cavallero2024co-existence, hoglund20243gpp}\footnote{This common assumption in the literature stems from the low achievable rate of \glspl{wus}. In practice, the duration also depends on the number of bits used to encode the \gls{id}, which is outside the scope of this work.}. The data transmission instead occupies the \gls{ul} portion of a slot, i.e., a duration of $\beta_t \tau$, with the payload assumed to fit entirely within this interval. Both $k_w$ and the data slot duration are configurable system parameters (e.g., via the modulation and coding scheme) and do not affect the validity of the model.}
Therefore, \textcolor{black}{a pull request and its reply last for} 
$T_w = (k_w + 1) \tau$ [s], and the time reserved \revise{for the pull sub-frame is} 
\begin{equation} \label{eq:tpull}
    T_{\rm pull} = Q T_w = Q (k_w + 1) \tau.
\end{equation}

The \textbf{push sub-frame} comprises a \gls{dl} \emph{control signal} of $k_c$ slots indicating the beginning of this part of the frame, and $P$ \emph{push slots}. The time dedicated \revise{for the push sub-frame} is
\begin{equation} \label{eq:tpush}
    T_{\rm push} = (P + k_c) \tau.
\end{equation}
All the nodes in an alarm state receiving the control signal will contend for access \textcolor{black}{to} the medium \textcolor{black}{in} the available $P$ push slots employing a \revise{random access scheme}.
\revise{For mathematical tractability, we employ the \emph{\gls{fsa}} protocol~\cite{wieselthier1989aloha}, but alternative schemes can be used with no significant change in the analysis.} In our case, the nodes in \textcolor{black}{the} alarm state will transmit their data in one out of the $P$ available push slots chosen uniformly at random. 

\revise{Focusing on \gls{mac} performance evaluation, we consider a collision channel with no capture. This implies that pull replies are \emph{error-free}, while push attempts are successful if and only if a single node transmits in one of the $P$ slots.\footnote{\revise{Errors related to \textcolor{black}{the} physical layer would affect both push and pull sub-frames in the same manner; thus, they are disregarded from the analysis.}}}

\begin{remark}
    Due to the fact that $T = T_\mathrm{pull} + T_\mathrm{push}$, the number of pull and push slots \textcolor{black}{is} related as follows:
    \begin{equation} \label{eq:pull-push-slots}
        Q = \left\lfloor \frac{F - P - k_c}{k_w + 1}\right\rfloor, \quad P = F - Q(k_w+1) - k_c.
    \end{equation}  
    \label{rem:duration}
\end{remark}

\begin{figure}[t!]
\begin{center}
    \includegraphics[width = .7\columnwidth]{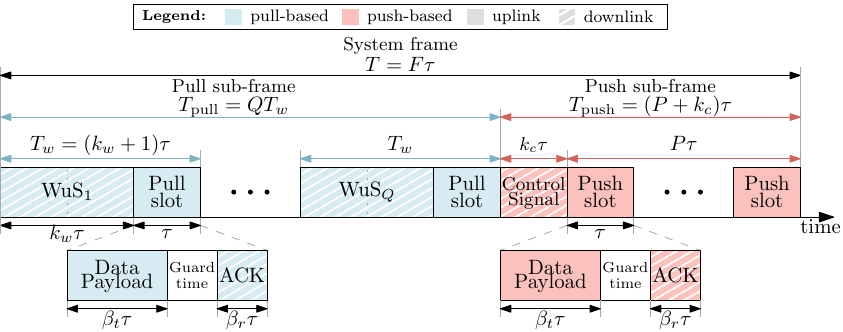}
    \caption{Time frame structure of pull and push-based communication.}
    \label{fig:time-diagram}
\end{center}
\end{figure}

\textcolor{black}{Note that} the control signal of the push sub-frame may need to embed a synchronization sequence and the frame structure parameters \revise{to let the nodes acquire frame information by periodically turning on the main radio to decode the control signal}. For simplicity, we further consider the following.
\begin{assumption}[Synchronization]
    All nodes are perfectly synchronized at the frame boundaries and know the frame structure perfectly, i.e., they know $F$, $Q$, $P$, $k_c$, and $\tau$.
    \label{assu:synch}
\end{assumption}
\noindent \revise{As periodic synchronization tasks are present in most classical push-based communication schemes, we do not consider the overhead and the energy of the synchronization process\footnote{Investigating the impact of clock drifts in the sensors, main radio, and wake-up receiver, as well as imperfections in the synchronization process, is out of the scope of this paper.}.}

\subsection{Queries}
\label{sec:model:queries}
The pull-based traffic is controlled by queries sent by the cloud to the \gls{bs}. 
The query includes the information of data collection deadline, by which the \gls{bs} must collect data from the specific device. 
If not served within the specified deadline, the query is discarded. To focus on an \gls{iot} data collection scenario and to simplify our analysis, this work focuses on a strict latency-constrained scenario, where the deadline is set to a single frame duration, i.e., queries requested at the beginning of a frame should be served by the end of the same frame.

\textcolor{black}{To this end}, we consider an arrival process having the properties specified by the following assumption. 
\begin{assumption}[Query generation]
   \revise{A maximum of one query per node per frame is generated, having an average device's arrival rate of $\lambda_q$ [queries/s], and resulting in a maximum of $N$ queries collected at the \gls{bs} per frame.} 
    \label{assu:queries}
\end{assumption}
Let us now denote the number of queries at the beginning of the generic $t$-th frame as $n_q(t)$. According to Assumption~\ref{assu:queries}, it is $n_q(t) = \min (\tilde{n}_q(t), N)$, 
\revise{where $\tilde{n}_q (t)$ is an unbounded \gls{rv} having mean value $\mu_q = N\, \lambda_q\, T$ [queries/frame] representing the average number of queries collected in a frame.}

%
According to Sec.~\ref{sec:model:frame}, only $Q$ \glspl{wus} related to the first $Q$ queries will be transmitted by the \gls{bs}. Therefore, the maximum number of served queries per frame is $\min\{Q, n_q(t)\}$. Nevertheless, the following consideration holds.
\begin{remark} \label{rem:query-on-alarm}
If a query contains the \gls{id} of a device in an alarm state, the device will receive the \gls{wus} and respond with information about the anomaly within the pull sub-frame \revise{due to its urgency}. This response will not satisfy the original query, as the device preempts the data related to the anomaly over the requested information. \revise{Note that, \textcolor{black}{the} \gls{bs} can distinguish alarm from regular query replies through decoding, e.g., via distinct headers or identifiers.} 
\end{remark}

\subsection{Anomalies}
\label{sec:model:anomalies}
Upon detecting an anomaly in its monitored process, \textcolor{black}{the device enters the alarm state} -- and thus it becomes an \emph{alarmed node} -- and turns on its main radio to transmit the updates containing information regarding the alarm state, called the \emph{alarm packet} in the remainder of the paper. The following assumption holds.
\begin{assumption}[Alarm state and packet generation]
    \revise{In each frame, each node may detect an anomaly according to a common discrete distribution with mean value $\lambda_a$ [packet/s]}; a node detecting an anomaly during frame $t-1$, will enter the alarm state from frame $t$, where it starts attempting data transmission; a device in an alarm state can generate only one alarm packet, limiting the number of alarm packets per device per frame to one; if the access is unsuccessful, the alarm packet is re-transmitted in the following frames -- either in a push or a pull manner, cf. Assumption~\ref{assu:alarm-tx} -- until it is successfully delivered to the \gls{bs}, whose \gls{ack} response terminates the alarm state of the node, as in~\cite{chiariotti2024distributed}. 
    \label{assu:alarm-state} 
\end{assumption}

According to Assumption~\ref{assu:alarm-state}, the number of nodes in the alarm state -- or, equivalently, the alarm packets in the system -- during frame $t$ is 
\begin{equation}
\label{eq:na-def}
n_a(t) = n_n(t-1) + n_f(t-1) \le N,     
\end{equation}
where $n_f(t-1)$ represents the alarm packets that have failed transmission within frame $t-1$, and $n_n(t-1)$ are the new alarm packets generated during frame $t-1$.
Moreover, $n_n(t)$ is upper-bounded by the number of nodes not in the alarm state, $N - n_a(t)$, $\forall t$. 
\revise{Therefore, the \gls{pmf} of $n_n(t)$ conditioned on the knowledge of $n_a(t)$--denoted as $n_n(t) | n_a(t)$--has a conditional mean of $\mu_a(t) = (N - n_a(t)) T \lambda_a$ [packet/frame], and depends on the system behavior \textcolor{black}{in} previous frames.} 
\revise{As an initial condition of the system,} it is reasonable to assume that the number of nodes in the alarm state is $0$ before the first frame, i.e., \revise{$n_a(0) = n_f(0) =0$. Thus, imposing the initial condition, we have $n_a(1) = n_n(0)$, whose \gls{pmf} is considered known, and having mean value $\mu_a(0) = N T \lambda_a$ [packet/frame].}
For the subsequent frames, an iterative analysis is necessary, as detailed in Sec.~\ref{sec:prob-analysis}.

The transmission of alarm packets occurs as follows.
\begin{assumption}[Transmission of alarm packets]
    \label{assu:alarm-tx}
    An alarmed node transmits its alarm packet \revise{either: 
    within the pull sub-frame, if the \gls{bs} transmits a \gls{wus} to the node, which in turn replies preempting the alarm packet to the requested query (c.f. Remark~\ref{rem:query-on-alarm}); or within the push sub-frame, attempting transmission following \textcolor{black}{an} \gls{fsa} approach.}
\end{assumption}
To model the system according to Assumption~\ref{assu:alarm-tx}, we denote the number of alarmed nodes receiving a \gls{wus} within frame $t$ as $n_w(t)$. Due to the error-free pull-based communication assumptions, $n_w(t)$ also represents the number of successfully transmitted alarm packets within the pull sub-frame. Therefore, the number of alarm packets attempting access during the push sub-frame of frame $t$ is
\begin{equation}
    \label{eq:np-def}
    n_p(t) = n_a(t) - n_w(t).
\end{equation}
Finally, we can relate $n_p(t)$ with the alarm packets unsuccessfully transmitted, $n_f(t)$, obtaining the alarm packets successfully transmitted in the push sub-frame of frame $t$ as 
\begin{equation}
    \label{eq:ns-def}
    n_s(t) = n_p(t) - n_f(t).
\end{equation}
The variables relations are illustrated in Fig.~\ref{fig:packets}.

\begin{figure}[!t]
    \centering
    \includegraphics[width=0.7\columnwidth]{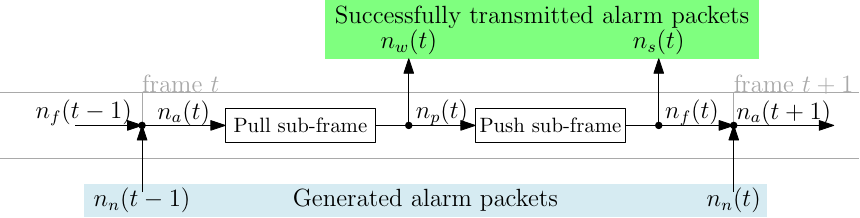}
    \caption{Relationship between generated and successfully transmitted alarms throughout the frame.}
    \label{fig:packets}
\end{figure}

\subsection{Energy consumption}
\label{sec:model:energy}

We assume that \gls{wur} operations are \emph{synchronous} with the \gls{mac} frame, according to Assumption~\ref{assu:synch}. In other words, the nodes can activate/deactivate the \gls{wur} taking advantage of the frame structure, by, e.g., turning off the \gls{wur} in the push sub-frame or when the main radio is activated. 
Hence, to model the energy consumption, we consider the following three terms representing the power consumption of a device: $\xi_w$ [W] is the power spent when the \gls{wur} is turned-on, while $\xi_r$ [W] and $\xi_t$ [W] denote the power consumed when the main radio is activated for data reception and transmission, respectively. The power consumption of the \gls{wur} is several orders of magnitude lower than that of the main radio, i.e., it satisfies $\xi_w \ll \xi_r \le \xi_t$. 
When both the main radio and the \gls{wur} are turned off, only the always-on sensor consumes energy. However, we ignore this quantity, as it is a constant factor independent from the network design.
Accordingly, the energy consumed by a device within a frame depends on the radio components activated, which, in turn, depend on the node's state (alarmed or not) and whether it receives a \gls{wus}. Through Assumptions~\ref{assu:synch},~\ref{assu:alarm-state}, and Remark~\ref{rem:query-on-alarm}, we can enumerate the following three cases, visualized in Fig.~\ref{fig:power-plot}.

\begin{enumerate}[label=Case~\arabic*), start=1, left=0pt]
    \item \emph{Pulled nodes}: The node (regardless \textcolor{black}{of} its state) receives a \gls{wus}. 
    Hence, the \gls{wur} is active and consumes power $\xi_w$ from the beginning of the frame until the end of the \gls{wus} intended for the node; detecting the \gls{wus} embedding the node's \gls{id}, the \gls{wur} triggers the activation of the main radio to communicate in the subsequent pull slot, consuming $\xi_t$ for payload transmission and $\xi_r$ for \gls{ack} reception (cf. Fig.~\ref{fig:time-diagram}); finally, the radio components are deactivated until the end of the frame.

    \item \emph{Push due to alarm state}: The node is in an alarm state; it does not receive any \gls{wus} within the pull sub-frame. 
    The node's \gls{wur} is active to listen for communication requests for the whole pull sub-frame, consuming power $\xi_w$; in \textcolor{black}{the} absence of any \gls{wus}, the node attempts data transmission within the push sub-frame while turning off its \gls{wur}. Within the push sub-frame, the main radio is active for the control signal reception, spending $\xi_r$; then, it activates for payload transmission and \gls{ack} reception in the selected push slot, consuming $\xi_t$, and $\xi_r$, respectively; for the other push slots, the main radio is deactivated. 

    \item \emph{Neither alarm \textcolor{black}{nor} pulled}: the node is not in an alarm state and does not receive a \gls{wus}. Hence, the node activates its \gls{wur} for the pull sub-frame -- consuming $\xi_w$ -- and deactivates it in the push sub-frame.
\end{enumerate}
The network energy consumption depends on the statistics of the queries and anomalies, analyzed next.

\begin{figure}[!t]
    \centering
    \includegraphics[width=.7\columnwidth]{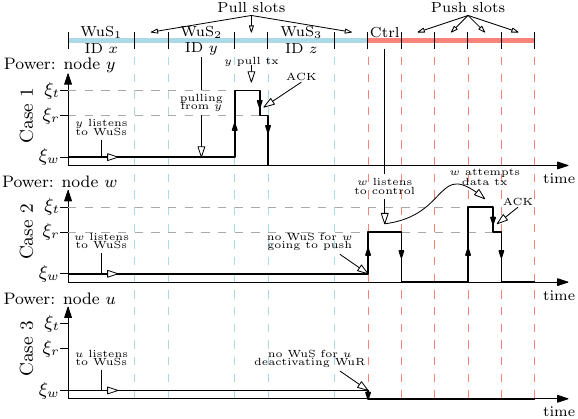}
    \caption{Diagram of the power spent by three nodes ($y$, $w$, $u$) representing the three energy cases analyzed, respectively. $y$ (case 1) receives the 2nd \gls{wus} and replies with a transmission on the subsequent pull slots. $w$ (case 2) is in the alarm state; not receiving any \gls{wus}, it listens to the push control signal and \textcolor{black}{decides} to attempt data transmission on the 3rd push slot; $u$ (case 3) turns off its radio components in the push sub-frame, not having received any \gls{wus}.}
    \label{fig:power-plot}
\end{figure}

\section{Performance Analysis}
\label{sec:performance}
In this section, we conduct a mathematical analysis of the system performance in terms of the \emph{success probability} of anomalies and queries (Sec.~\ref{sec:prob-analysis}), and in terms of the \emph{overall energy consumed} by the nodes (Sec.~\ref{sec:energy-analysis}). The aim is \textcolor{black}{to link} the performance with the frame structure defined in Sec.~\ref{sec:model:frame}, focusing on the relative durations of pull and push sub-frames, controlled by the parameter $Q$--cf. Remark~\ref{rem:duration}.

\subsection{Performance of reporting anomalies and serving queries}
\label{sec:prob-analysis}
We begin by providing the recursive relations between the alarmed nodes in two subsequent frames, used to derive the success probability of \emph{transmitting the alarm packets}, with frame $t$ representing the starting point of our evaluation. Then, we present the success probability of \emph{satisfying the queries} received by the \gls{bs}. Finally, the previous quantities are translated into a unique metric, denoted as a \emph{success probability trade-off}, which characterizes the success probabilities of alarms and queries as a function of $Q$.

\subsubsection{Recursive analysis of the alarm packets}
\label{sec:prob-analysis:anomalies}
In this section, we find the \gls{pmf} of the number of alarm packets in frame $t+1$, $n_a(t+1)$, as a function of the \gls{pmf} of $n_a(t)$.

Based on the statistics of $n_a(t)$\footnote{The analysis is valid considering the initial condition \revise{$n_a(1) = n_n(0)$}.}, we proceed to analyze the \gls{pmf} of $n_w(t)$, which is influenced by the number of generated alarm packets $n_a(t)$ and the number of queries received $n_q(t)$. Through the law of total probability, we have
\begin{equation}  
    \mathcal{P}(n_w(t) = j) = \sum_{i=0}^N \sum_{q=0}^N \mathcal{P}(n_w(t) = j|n_q(t) = q, n_a(t) = i)  \mathcal{P}(n_q(t) = q)\, \mathcal{P}(n_a(t) = i),
    \label{eq: prob(n_w=j)}
\end{equation}
where the expression $\mathcal{P}(n_w(t) = j|n_q(t) = q, n_a(t) = i)$ denotes the probability that, given $q$ queries received, exactly $j$ nodes out of $i$ nodes with alarms receive the \gls{wus}, given $N$, and $Q$. 
\revise{The previous conditioned \gls{pmf} depends on the \gls{pmf} of the number of queries generated, $n_q(t)$, assumed to be known by the \gls{bs} and tunable at the cloud level. An example of this \gls{pmf} for a simple solution is given in Sec.~\ref{sec:results}.}
To simplify the notation, we denote
\begin{equation} 
\begin{split}
\label{eq:nw|na}
    &\mathcal{P}(n_w(t) = j|n_a(t) = i) = \\
    &=\sum_{q=0}^N {\mathcal{P}(n_w(t) = j|n_q(t) = q, n_a(t) = i)\, \mathcal{P}(n_q(t) = q)}.
\end{split}
\end{equation}

Once the statistic of the alarmed nodes receiving a \gls{wus} is known, we can analyze the remaining number of nodes in the alarm state, $n_p(t)$, which do not receive the \gls{wus} during the pull sub-frame and will attempt to access the channel within the push sub-frame. According to~\eqref{eq:np-def}, 
the \gls{pmf} of $n_p(t)$ can be expressed as
\begin{equation}
    \mathcal{P}(n_p(t) = k) = \sum_{i=k}^{N} \mathcal{P}(n_w(t) = i-k | n_a(t) = i) \, \mathcal{P}(n_a(t) = i).
    \label{eq: prob_np=k}
\end{equation}

Of these $n_p(t)$ \textcolor{black}{alarmed nodes trying to} access the push sub-frame, some are successful and others fail because of collisions. Following~\eqref{eq:ns-def}, we can define the \gls{pmf} of $n_s(t)$ as
\begin{equation}
    \mathcal{P}(n_s(t) = s) = \sum_{k=0}^{N}{\mathcal{P}(n_s(t) = s | n_p(t) = k) \, \mathcal{P}(n_p(t) = k)},
    \label{eq: pmf_ns}
\end{equation}
where the \gls{pmf} of $n_s(t) | n_p(t)$ is given in the following Lemma.

\begin{lemma}
\label{lemma:alarm-success-given-attempts}
The probability of alarm packets successfully transmitted in the push sub-frame given the number of alarms trying to access can be obtained by computing the following marginal \gls{pmf}
\begin{equation}
 \mathcal{P}(n_s(t) = s | n_p(t) = k) \hspace{-1mm} = \hspace{-1mm} \sum_{c=0}^{N} \mathcal{P}(s, c | n_p(t) = k),
    \label{eq: pmf_ns_nf_np}  
\end{equation}
where the joint \gls{pmf} $\mathcal{P}(s, c | n_p(t))$ is given in~\eqref{eq:P(s,c|np)}.
\end{lemma}
\begin{proof}
    Please, refer to Appendix~\ref{sec:proof-alarm-success}.
\end{proof}

Moreover, the \gls{pmf} of the number of alarm packets failing the access in the push sub-frame, $n_f(t)$, can be obtained by using~\eqref{eq:ns-def} and applying the total law of probability, obtaining
\begin{equation}
    \mathcal{P}(n_f(t) = f) \hspace{-1mm} = \hspace{-1.8mm}\sum_{k=f}^{N} \mathcal{P}(n_s(t) = k - f | n_p(t) = k) \mathcal{P}(n_p(t) = k).
   \label{eq: pmf_nf}
\end{equation}

Using~\eqref{eq:na-def}, we obtain the \gls{pmf} of $n_a(t+1)|n_a(t)$ as
\begin{equation}
\begin{split}
    &\mathcal{P}(n_a(t+1) = a | n_a(t) = i) = \\
    &=\hspace{-1.8mm} \sum_{f=f_{\min}}^{a}\hspace{-1.5mm} \mathcal{P}(n_n(t) = a - f | n_a(t) = i)\mathcal{P}(n_f(t) = f | n_a(t) = i),
    \label{eq: pmf_ns2}
    \end{split}
\end{equation}
where $f_{\min} = \max\{0, a + i - N\}$ obtained from the conditions $0 \le a - f \le N - i$ and $f \ge 0$.
Employing the law of total probability on eq.~\eqref{eq: pmf_ns2}, we can finally write the \gls{pmf} of $n_a(t+1)$ as a function of all the previous \glspl{pmf}:
\begin{equation}
\begin{split}
    \mathcal{P}(n_a(t+1) = a) =& \sum_{i=0}^N \sum_{f=f_{\min}}^{a}   \mathcal{P}(n_n(t) = a - f | n_a(t) = i) \\
    &\sum_{k=f}^N  \mathcal{P}(n_s(t) = k - f | n_p(t) = k) \mathcal{P}(n_w(t) = i - k  | n_a(t) = i) \, \mathcal{P}(n_a(t) = i),
    \end{split}
    \label{eq:na(t+1)}
\end{equation}    
We remark that eq.~\eqref{eq:na(t+1)} is valid regardless of the \gls{pmf} of $n_n(t)|n_a(t)$.

To summarize, we have obtained the \gls{pmf} of the alarm packets in frame $t+1$ as a function of the \gls{pmf} of the alarm packets in frame $t$. Imposing the initial condition $n_a(0) = 0$, eq.~\eqref{eq:na(t+1)} can be used to obtain the \glspl{pmf} of the alarmed \textcolor{black}{nodes} in the system for all the frames $t\in\{1, 2, \dots\}$.

\subsubsection{Success probability of alarm packets}
\revise{We can compute the alarm success probability at frame $t$, denoted as $p_a(t)$.
According to Assumption~\ref{assu:alarm-tx}, an alarm packet is successfully delivered if it is transmitted either: 1) within the pull sub-frame, obtaining a \emph{pull success}, or 2) while accessing the push sub-frame, obtaining a \emph{push success}. The absence of alarm packets, i.e., $n_a(t) = 0$,  is also considered a success, because no anomaly has to be reported in the frame.}
\revise{Denoting the probability of a pull success in frame $t$ as $p_w(t)$, and the probability of a push success in frame $t$ as $p_p(t)$, we obtain}
\begin{equation} \label{eq:pa(t)}
    \revise{p_a (t) = \mc{P}(n_a(t) = 0) + p_w (t) + p_p (t)},
\end{equation}
\revise{where the events represented by the three terms are disjoint.}

\revise{Assuming $n_a(t) \ge 1$, a pull success occurs when a node in the alarm state receives a \gls{wus}.}
\textcolor{black}{Assuming that $n_w(t)$ and $n_a(t)$ are known}, the pull success probability is the ratio $n_w(t) / n_a(t)$. Thus, $p_w(t)$ is the previous ratio multiplied by the probability of having exactly $n_w(t)$ out of $n_a(t)$ alarmed nodes receiving a \gls{wus} in frame $t$, i.e.,
\begin{equation} \label{eq:pw(t)}
    p_w(t) = 
        \revise{\sum_{i=1}^N  \sum_{w=1}^{i} {\frac{w}{i}
    \mathcal{P}(n_w(t) = w|n_a(t) = i)\, \mathcal{P}(n_a(t) = i)},}
\end{equation}
\revise{which also depends on $n_q(t)$--see eq.~\eqref{eq:nw|na}.}

\revise{From the alarmed node's point of view, a push success occurs whenever a pull success did not occur, i.e., with probability $1 - n_w(t) /n_a(t)$, and no collision occurs on the $P$ available resources.}
\revise{In other words,
\begin{equation}\label{eq:pp(t)}
\begin{aligned}
    p_p(t) &= \sum_{i=1}^N \sum_{w=0}^{i-1} \left(1 - \frac{w}{i}\right) p_\mathrm{nc}(i-w) \\
    &\qquad\cdot\mc{P}(n_w(t) = w|n_a(t) = i) \mc{P}(n_a(t) = i),
\end{aligned}
\end{equation}
\textcolor{black}{where, with $k=i-w$,} $p_\mathrm{nc}(k)$ denotes the probability of no collision of a packet given a total of $k>0$ \textcolor{black}{packets} attempting access.
Eq.~\eqref{eq:pp(t)} is a general formulation that does not depend on the specific form of $p_\mathrm{nc}(k)$.}
Since our push sub-frame employs a \gls{fsa} protocol, \revise{$p_\mathrm{nc}(k)$ is}~\cite{israel_math}
\begin{equation} 
    p_\mathrm{nc}(k) =
    \begin{cases}
    0, 
    &\text{if } P=0 \lor P = 1 \land k > 1, \\
    \left( 1 - \frac{1}{P} \right)^{k - 1}, 
    &\text{if } P > 1 \land k \ge 1,    \\
    1, &\text{if } P = 1 \land k=1. 
\end{cases}
    \label{eq:ps_p}
\end{equation}
Through eq.~\eqref{eq:pa(t)}, it is possible to compute the success probability of alarm packets in any frame, knowing all the \glspl{pmf}. We further compute an average success probability of alarms, $\bar{p}_a$, to provide a simple performance metric evaluated on a reference period of observation $T_O$ (in frames), i.e.,
\begin{equation}
    \bar{p}_a = \frac{1}{T_O}\sum_{t=0}^{T_O-1} p_a(t).
\end{equation}

\subsubsection{Success probability of serving queries}
Here, we evaluate the success probability of serving queries, denoted as $p_q(t)$. Following Remark~\ref{rem:query-on-alarm}, a query is successfully served when a device that is not in the alarm state receives a \gls{wus}. Accordingly, we obtain
\begin{equation}
   \begin{split}
p_q(t) &= \sum_{i = 0}^N \sum_{j = 0}^i \sum_{q = 1}^N \frac{X - j}{q} \mathcal{P} (n_w(t) = j|n_q(t) = q, n_a(t) = i) \\
&\cdot \mathcal{P} (n_q(t) = q) \mathcal{P} (n_a(t) = i),
\end{split}
\end{equation}
where $X = \min\{q, Q\}$. Note that $\frac{X-j}{q}$ represents the ratio between the number of successfully served queries and the total number of received queries--assuming that at least one query is received, otherwise the probability of success is zero. This ratio is multiplied by the probability of generating $q$ queries, the probability of having $j$ alarm nodes receiving a \gls{wus}, and the one of having $i$ alarmed nodes in the system.

We further define the mean success probability of serving the queries averaging $p_q(t)$ over an observation period $T_O$ as
\begin{equation}
    \bar{p}_q = \frac{1}{T_O}\sum_{t=0}^{T_O-1} p_q(t).
\end{equation}

\subsubsection{Success probability trade-off}
Instead of jointly maximizing both alarms and queries success probabilities, we characterize a performance trade-off defining the weighted average success probability for the push/pull coexistence as
\begin{equation}
    \bar{p}_{s} = w_q\, \bar{p}_q + w_a \, \bar{p}_a,
\end{equation}
where $w_q~\in~[0, 1]$ and $w_{a}\in~[0, 1]$ are positive weights useful \textcolor{black}{for targeting the} requirements of both communication modes, respectively. 
Since both $\bar{p}_q$ and $\bar{p}_a$ are functions of $Q$, it is possible to find the value of $Q$ that maximizes $\bar{p}_{s}$, under a specific set of weights and system settings.
As a showcase, we select $w_q=\frac{\lambda_q}{\lambda_q + \lambda_a}$ and $w_a=\frac{\lambda_a}{\lambda_q + \lambda_a}$, considering fairness in terms of traffic load for both alarm packets and queries.

\subsection{Energy analysis}
\label{sec:energy-analysis}
In this section, we conduct an analysis of the average energy consumed by the system according to the three different cases defined in Sec.~\ref{sec:model:energy}. 

\subsubsection{Pulled nodes}
The average energy consumed in the frame $t$ by nodes (either alarmed or not) receiving a \gls{wus} and transmitting during the pull sub-frame is 
\begin{equation}
   \begin{split}
    E_1(t) &= \hspace{-1mm} \sum_{u = 1}^{\min(N, Q)} \hspace{-1mm} \mathcal{P}(n_u(t) = u) \, \tau \sum_{i = 1}^u \big[ \xi_w (i k_w + i-1)  \\
    &\qquad +  \xi_t \beta_t + \xi_r \beta_r \big],
    \end{split}
    \label{eq: E1}
\end{equation}
where $\beta_t$ and $\beta_r$ are the fractions of the slot duration dedicated to data transmission and \gls{ack} reception, respectively (cf. Sec.~\ref{sec:model:frame}), and $\mathcal{P} (n_u(t) = u)$ is the \gls{pmf} of having $u$ nodes, alarmed or not, receiving the \gls{wus} during the pull sub-frame. By definition, $n_u(t) = \min (Q, n_q(t))$, being $Q$ the maximum number of queries that can be served in a frame. Therefore: 
\begin{equation}
    \mathcal{P}(n_u(t) = u) = 
    \begin{cases}
        \mathcal{P}(n_q(t) = u), &\text{if } u < Q, \\
        \sum_{q=Q}^N \mathcal{P}(n_q(t) = q), &\text{if } u = Q.        
    \end{cases}
    \label{eq:prob_u_WuS}
\end{equation}

Specifically, $E_1(t)$ accounts for all possible combinations where the generic $i$-th node keeps its \gls{wur} powered on and its main radio off until it receives the $i$-th \gls{wus}. Afterward, it turns on its main radio to transmit its data packet and, upon receiving the \gls{ack}, switches off its radio components.

\subsubsection{Push due to alarm state}
The average energy consumed in the frame $t$ by the alarmed nodes attempting data transmission within the push sub-frame is 
\begin{equation}
   \begin{split}
E_2(t) &= \sum_{k = 1}^N \mathcal{P} (n_p(t) = k) \, \tau k \, [ \xi_w Q (k_w + 1) + \xi_r k_c + \xi_t \beta_t + \xi_r \beta_r],
\end{split}
\label{eq: E2}
\end{equation}
where $\mathcal{P}(n_p(t) = k)$ is given in~\eqref{eq: prob_np=k}. This energy consumption corresponds to the nodes that are alarmed, keeping their \glspl{wur} in reception mode throughout the pull sub-frame. After receiving the push control signal, these nodes attempt data transmission in one randomly chosen slot of the push sub-frame, while they are deactivated for the remaining $P - 1$ slots.

\subsubsection{Neither alarmed \textcolor{black}{nor} pulled}
The energy consumed in frame $t$ by the nodes that are not alarmed and do not receive a \gls{wus} is 
\begin{equation}   
        E_3(t) = \sum_{y = 1}^N \mathcal{P} (n_y(t) = y)\, \tau  y \, [ \xi_w Q (k_w + 1)],
    \label{eq: E3}
\end{equation}
where $n_y(t)$ is the number of nodes that are not alarmed and do not receive any \gls{wus} during the pull sub-frame. \textcolor{black}{These nodes keep the \gls{wur} turned on} during the whole pull sub-frame and then deactivate all the radio components. The \gls{pmf} of $n_y(t)$ is given in the following Lemma.

\begin{lemma}
    \label{lemma:pmf_ny}
    The \gls{pmf} of the number of not alarmed nodes not receiving any \gls{wus} is 
    \begin{equation}
       \begin{split}
        \mathcal{P} (n_y(t) = y) &=  \sum_{q = 0}^{N} \sum_{j = 0}^{N} \sum_{i = 0}^{N} \sum_{z=0}^N \mathcal{P} (n_y(t) = y| n_z(t) = z, n_a(t) = i) \, \mathcal{P} (n_z(t) = z| n_w(t) = j, n_q(t) = q) \\ &\qquad\cdot \mathcal{P} (n_w(t) = j | n_q(t) = q, n_a(t) = i) \mathcal{P} (n_a(t) = i) \mathcal{P} (n_q(t) = q).
    \end{split}
    \label{eq: prob_ny}
    \end{equation}    
    where $n_z(t)$ is defined as the number of non-alarmed nodes receiving a \gls{wus}. The terms in~\eqref{eq: prob_ny} are given by~\eqref{eq:queries-pmf},~\eqref{eq: prob(n_w=j)},~\eqref{eq:na(t+1)},~\eqref{eq:prob_z_unalarmed_WuS}, and~\eqref{eq: prob(n_y=y|nq, na)}.
\end{lemma}
\begin{proof}
    Please, refer to Appendix~\ref{sec:proof-pmf_ny}.
\end{proof}

\subsubsection{Average energy consumed by the system}
 The system's energy consumption under different conditions can be computed recursively through eqs.~\eqref{eq: E1},~\eqref{eq: E2}, and~\eqref{eq: E3}. 
The average energy consumption, $\bar{E}$, is evaluated considering a reference observation period $T_O$ (in frames), i.e.,
\begin{equation}
    \bar{E} = \frac{1}{T_O}\sum_{t=0}^{T_O-1} E_1(t) + E_2(t) + E_3(t).
    \label{eq: E_wur}
\end{equation}

\section{Numerical Results}
\label{sec:results}
In this section, we numerically evaluate the performance of the proposed system. \revise{A custom simulator, implementing both pull and push communication modes and accounting for the simplifying assumptions of the analytical model, has been developed.}
In the simulator we set the power consumption of \gls{wur} to 1 [mW], which should be considered as a worst case, as the actual power consumption of \gls{wur} is \textcolor{black}{on} the order of several micro watts~\cite{piyare2017ultra}. The presented results are obtained by averaging over $10^5$ Monte Carlo simulation instances. Table~\ref{tab:system_parameter_settings_pull-push2} reports the parameter values used for the numerical results, obtained for a slot comprising 7 \gls{ofdm} symbols using 5G numerology 1, i.e., a subcarrier spacing of $30$ kHz. The chosen parameters lead to a system in which $Q\in[0, 8]$, where $Q = 0$ represents a frame comprising a push sub-frame only, while with $Q=8$ only the pull sub-frame is considered.

\paragraph*{Query modeling}
\revise{The (unbounded) number of queries collected at the \gls{bs}, $\tilde{n}_q(t)$, follows a Poisson distribution\footnote{\revise{Modeling queries as a Poisson process is a common assumption in \gls{3gpp} studies to capture uncoordinated, low-rate arrivals~\cite{3gpp_poisson}.}}. Thus, the \gls{pmf} of $n_q(t) = \min(\tilde{n}_q(t), N)$ is
\begin{equation}
\label{eq:queries-pmf}
    \mathcal{P}(n_q(t) = q) =
    \begin{cases}
       \frac{\mu_q^q  e^{-\mu_q}}{q!} , &\text{if } q < N, \\
        1 - \sum_{q=0}^{N-1}{\frac{\mu_q^q  e^{-\mu_q}}{q!}} , &\text{if } q = N, \\
        0, &\text{otherwise}, \\
    \end{cases}
\end{equation}
with $\mu_q = N \lambda_q T$. Moreover, the \glspl{id} of the queries generated in a frame are randomly sampled without replacement from the set $\{1,\dots, N\}$.
Therefore, the conditional \gls{pmf} in~\eqref{eq: prob(n_w=j)} yields:}
\begin{equation}
\begin{split}
    &\mathcal{P}(n_w(t) = j|n_q(t) = q, n_a(t) = i) = 
    \begin{cases}    
         \frac{\binom{i}{j} \, \binom{N-i}{X-j}}{\binom{N}{X}},  &\text{if } N \ge q \ge j \land i \geq j, \\         
         0, &\text{otherwise},
    \end{cases}
    \end{split}
    \label{eq: prob(n_w=j|nq,na)}
\end{equation}
\revise{where $X = \min(q, Q)$. 
Eq.~\eqref{eq: prob(n_w=j|nq,na)} is obtained considering that when $N \ge q \ge j$ and $i \geq j$, the denominator accounts for all possible combinations of selecting $X$ nodes out of the $N$ available ones. Meanwhile, the numerator represents the combinations where $j$ alarmed nodes receive a \gls{wus}, while the remaining $X - j$ nodes do not.
Note that these assumptions lead to worst-case performance that can be improved by more intelligent query schedulers, e.g.,~\cite{chiariotti2026combined}.} 

\paragraph*{Anomaly modeling}
\revise{We consider both uncorrelated and correlated anomaly generation across devices.}

\revise{For the uncorrelated case, in each frame, every device may independently detect an anomaly with probability $\lambda_a T$.
Thus, the number of newly generated alarm packets in frame $t$, conditioned on the number of alarmed nodes $n_a(t)$, follows a Binomial distribution, i.e., $n_n(t)\mid n_a(t) \sim \mathrm{B}(N-n_a(t),\lambda_a T)$, having mean $\mu_a(t)$--cf. Sec. \ref{sec:model:anomalies}.
Imposing the initial condition yields $n_n(0) = n_a(1)\sim\mathrm{B}(N,\lambda_a T)$.}

\revise{To capture correlations of anomalies due to shared external events or operating conditions, we also} \revise{adopt a Beta--Binomial formulation.
In each frame, all devices share a common random alarm activation probability} \revise{$p_c \sim \mathrm{Beta}(k \lambda_a  T, k (1 - \lambda_a T))$, so that $\mathbb{E}[p_c]=\lambda_a T$.} 
\revise{The concentration parameter is set to $k = 0.5$, allowing substantial frame-to-frame variability in $p_c$ and thereby inducing correlated, bursty alarm behavior. Conditioned on $p_c$, alarm events are independent across devices, and the number of newly generated alarms in frame $t$ satisfies $n_n(t)\mid n_a(t),p_c \sim \mathrm{B}(N-n_a(t),p_c)$.
Marginalizing over $p_c$ yields a Beta--Binomial distribution for $n_n(t)$ with support $0\le n_n(t)\le N-n_a(t)$, capturing positive correlation among alarm events.}

\begin{table}[t!]
\centering
\footnotesize
\def\arraystretch{1.2}
\caption{Simulation parameters.}
\vspace{-5pt}
\label{tab:system_parameter_settings_pull-push2}
\begin{tabular}{|l|l|l|l|}
\hline
\textbf{Param.} & \textbf{Value} & \textbf{Param.} & \textbf{Value} \\\hline
$N$ & 40 &  $k_c$, $k_w$ & 1, 4~\cite[Table 7.1.2.2-5]{3gpp:38-869}\\
$\tau$ & 0.25 ms & $\xi_w$, $\xi_t$, $\xi_r$  & 1, 55, 50~mW~\cite{tamura2019low_all}\\
$F$ & 41 & $\beta_t$, $\beta_r$ & 4/7, 3/7 \\
$T$  &  10.25 ms & $T_O$ & 10 \\
\hline
\end{tabular}
\end{table}

\subsection{Analytical model validation}

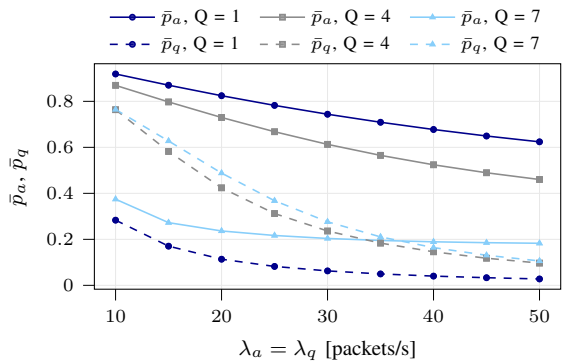
\begin{figure}[t]
\centering
\vspace{-2mm}
    \subfloat{\begin{tikzpicture}
\begin{axis}[
    width=0cm,
    height=0cm,
    axis line style={draw=none},
    tick style={draw=none},
    at={(0,0)},
    scale only axis,
    xmin=0,
    xmax=1,
    xtick={},
    ymin=0,
    ymax=1,
    ytick={},
    legend cell align={left},
    legend style={at={(0.5, 0)}, anchor=center, draw=none, fill=none, /tikz/every even column/.append style={column sep=5pt}},
    legend columns = -1,
    ]

\addlegendimage{darkblue, only marks, mark=*, mark options={fill=darkblue}, mark size=2}
\addlegendentry{$Q = 1$}
\addlegendimage{gray, only marks, mark=square*, mark options={fill=gray}, mark size=2}
\addlegendentry{$Q = 4$}
\addlegendimage{lightskyblue, only marks, mark=triangle*, mark options={fill=lightskyblue}, mark size=2,}
\addlegendentry{$Q = 7$}
\addlegendimage{black, semithick}
\addlegendentry{$\bar{p}_a$}
\addlegendimage{black, dashed, semithick}
\addlegendentry{$\bar{p}_q$}
\addlegendimage{black, dashed, dash pattern=on 5pt off 3pt on 1pt off 3pt, semithick}
\addlegendentry{$\bar{E}$}

\end{axis}
\end{tikzpicture}}\\
    \vspace{-5mm}
\setcounter{subfigure}{0}
    \subfloat[Avg. success prob., uncorrelated \label{fig:analytical_alarm_query}]{
        \centering
\begin{tikzpicture}

\begin{axis}[
scale only axis,
width=\four,
height=\four,
xlabel={\(\displaystyle \lambda_a = \lambda_q\) [packets/s]},
xmajorgrids,
xmin=8, xmax=52,
ylabel={\(\displaystyle \bar{p}_a\), \(\displaystyle \bar{p}_q\)},
ymajorgrids,
ymin=0, ymax=1,
]
\addplot [draw=darkblue, fill=darkblue, forget plot, mark=*, mark size=1, only marks]
table{%
x  y
10 0.918899204
15 0.870334566
20 0.824326266
25 0.782394092
30 0.743766355
35 0.708772819
40 0.677646078
45 0.649023226
50 0.624070859
};
\addplot [draw=gray, fill=gray, forget plot, mark=square*, mark size=1, only marks]
table{%
x  y
10 0.869650544
15 0.797343137
20 0.729601043
25 0.667850787
30 0.612656904
35 0.564696283
40 0.524028478
45 0.489575904
50 0.46011553
};
\addplot [draw=lightskyblue, fill=lightskyblue, forget plot, mark size=1, mark=triangle*, only marks]
table{%
x  y
10 0.37445486
15 0.272896037
20 0.236244644
25 0.216597496
30 0.203983988
35 0.195102632
40 0.189358093
45 0.185571269
50 0.182708439
};
\addplot [draw=darkblue, fill=darkblue, forget plot, mark=*, mark size=1, only marks]
table{%
x  y
10 0.283251178
15 0.170330479
20 0.1136267
25 0.082194776
30 0.062738909
35 0.049570974
40 0.040204327
45 0.03311216
50 0.027728487
};
\addplot [draw=gray, fill=gray, forget plot, mark=square*, mark size=1, only marks]
table{%
x  y
10 0.763341567
15 0.581940765
20 0.423946589
25 0.312177751
30 0.236242648
35 0.183537573
40 0.145863203
45 0.117890386
50 0.096746895
};
\addplot [draw=lightskyblue, fill=lightskyblue, forget plot, mark=triangle*, mark size=1, only marks]
table{%
x  y
10 0.764175312
15 0.627430909
20 0.488090669
25 0.367564675
30 0.276171733
35 0.210372783
40 0.164040726
45 0.13052923
50 0.106243719
};
\addplot [semithick, darkblue, mark=*, mark size=1, mark options={solid}]
table {%
10 0.918964744
15 0.870133266
20 0.824464347
25 0.782387997
30 0.743950714
35 0.709071742
40 0.67759175
45 0.649296664
50 0.623938581
};
\addplot [semithick, gray, mark=square*, mark size=1, mark options={solid}]
table {%
10 0.869468203
15 0.797466951
20 0.729831973
25 0.66791767
30 0.612896854
35 0.565082187
40 0.524152377
45 0.489442172
50 0.460147318
};
\addplot [semithick, lightskyblue, mark=triangle*, mark size=1, mark options={solid}]
table {%
10 0.374542264
15 0.272543371
20 0.236426758
25 0.216772981
30 0.203828618
35 0.195124395
40 0.189324799
45 0.18546051
50 0.182859306
};
\addplot [semithick, darkblue, dashed, mark=*, mark size=1, mark options={solid}]
table {%
10 0.283430632
15 0.170370217
20 0.113590568
25 0.08220138
30 0.062737987
35 0.049589394
40 0.040170687
45 0.033142877
50 0.027740889
};
\addplot [semithick, gray, dashed, mark=square*, mark size=1, mark options={solid}]
table {%
10 0.763702584
15 0.582167541
20 0.423885631
25 0.31218025
30 0.236241816
35 0.183573199
40 0.145810258
45 0.117894607
50 0.096739773
};
\addplot [semithick, lightskyblue, dashed, mark=triangle*, mark size=1, mark options={solid}]
table {%
10 0.764393106
15 0.6275908
20 0.488027113
25 0.367864817
30 0.276268627
35 0.210609492
40 0.164112285
45 0.130696874
50 0.106086695
};
\end{axis}

\end{tikzpicture}
    }
    \qquad
    \subfloat[Avg. energy, uncorrelated\label{fig:E_vs_lambda}]{
\begin{tikzpicture}

\begin{axis}[
scale only axis,
width=\four,
height=\four,
xlabel={\(\displaystyle \lambda_a = \lambda_q\) [packets/s]},
xmajorgrids,
xmin=8, xmax=52,
ylabel={\(\displaystyle \bar{E}\) [mJ]},
ymajorgrids,
ymin=0.12, ymax=1.08,
]
\addplot [draw=darkblue, fill=darkblue, forget plot, mark=*, mark size=1, only marks]
table{%
x  y
10 0.165806893
15 0.216843772
20 0.266953074
25 0.315491
30 0.362203022
35 0.406680766
40 0.44849762
45 0.488143854
50 0.524813634
};
\addplot [draw=gray, fill=gray, forget plot, mark=square*, mark size=1, only marks]
table{%
x  y
10 0.338068616
15 0.39599586
20 0.451040382
25 0.504749646
30 0.557055781
35 0.606755551
40 0.652992921
45 0.695746782
50 0.734759968
};
\addplot [draw=lightskyblue, fill=lightskyblue, forget plot, mark=triangle*, mark size=1, only marks]
table{%
x  y
10 0.579869458
15 0.68826432
20 0.767091778
25 0.82895169
30 0.879862632
35 0.922343802
40 0.957546917
45 0.987391791
50 1.013128772
};
\addplot [semithick, darkblue, forget plot, dashed, dash pattern=on 5pt off 3pt on 1pt off 3pt,] 
table {%
10 0.165812367
15 0.216922755
20 0.26689897
25 0.315398078
30 0.362030983
35 0.406479714
40 0.44851992
45 0.488022335
50 0.52494377
};
\addplot [semithick, gray, forget plot, dashed, dash pattern=on 5pt off 3pt on 1pt off 3pt,] 
table {%
10 0.338102088
15 0.396049948
20 0.450969155
25 0.504675518
30 0.55680497
35 0.606470719
40 0.652931623
45 0.69575406
50 0.734796534
};
\addplot [semithick, lightskyblue, forget plot, dashed, dash pattern=on 5pt off 3pt on 1pt off 3pt,] 
table {%
10 0.579795758
15 0.688516842
20 0.766862301
25 0.828808495
30 0.879833088
35 0.922218908
40 0.957631476
45 0.987506815
50 1.013011069
};

\end{axis}

\end{tikzpicture}
    }
    \\
    \vspace{-3mm}
    \subfloat[Avg. success prob., correlated\label{fig:pq_pa_correlated}]{
\begin{tikzpicture}

\begin{axis}[
scale only axis,
width=\four,
height=\four,
xlabel={\(\displaystyle \lambda_a = \lambda_q\) [packets/s]},
xmajorgrids,
xmin=8, xmax=52,
ylabel={\(\displaystyle \bar{p}_a\), \(\displaystyle \bar{p}_q\)},
ymajorgrids,
ymin=0, ymax=1,
]
\addplot [draw=darkblue, fill=darkblue, forget plot, mark=*, mark size=1, only marks]
table{%
x  y
10 0.88955457
15 0.843079395
20 0.80197418
25 0.763851782
30 0.730389057
35 0.699316814
40 0.671383146
45 0.646436942
50 0.623291475
};
\addplot [draw=gray, fill=gray, forget plot, mark=square*, mark size=1, only marks]
table{%
x  y
10 0.824283837
15 0.764096562
20 0.706867862
25 0.656354824
30 0.609966981
35 0.569317038
40 0.534011669
45 0.501790491
50 0.474420362

};
\addplot [draw=lightskyblue, fill=lightskyblue, forget plot, mark size=1, mark=triangle*, only marks]
table{%
x  y
10 0.546894171
15 0.458649598
20 0.39855995
25 0.34927709
30 0.309771131
35 0.279357099
40 0.255500787
45 0.237450567
50 0.22523034
};
\addplot [draw=darkblue, fill=darkblue, forget plot, mark=*, mark size=1, only marks]
table{%
x  y
10 0.269871788
15 0.160029781
20 0.105513731
25 0.075790195
30 0.057724052
35 0.045684087
40 0.03706195
45 0.0307127
50 0.02588665
};
\addplot [draw=gray, fill=gray, forget plot, mark=square*, mark size=1, only marks]
table{%
x  y
10 0.707007585
15 0.530140861
20 0.381037381
25 0.279523066
30 0.211543445
35 0.165381606
40 0.132546618
45 0.108041709
50 0.089874093
};
\addplot [draw=lightskyblue, fill=lightskyblue, forget plot, mark=triangle*, mark size=1, only marks]
table{%
x  y
10 0.679938452
15 0.582100886
20 0.467624396
25 0.359890593
30 0.273716302
35 0.210166652
40 0.164274353
45 0.130783691
50 0.106644331
};
\addplot [semithick, darkblue, mark=*, mark size=1, mark options={solid}]
table {%
10 0.88955457
15 0.843079395
20 0.80197418
25 0.763851782
30 0.730389057
35 0.699316814
40 0.671383146
45 0.646436942
50 0.623291475
};
\addplot [semithick, gray, mark=square*, mark size=1, mark options={solid}]
table {%
10 0.824283837
15 0.764096562
20 0.706867862
25 0.656354824
30 0.609966981
35 0.569317038
40 0.534011669
45 0.501790491
50 0.474420362
};
\addplot [semithick, lightskyblue, mark=triangle*, mark size=1, mark options={solid}]
table {%
10 0.546894171
15 0.458649598

20 0.39855995
25 0.34927709
30 0.309771131
35 0.279357099
40 0.255500787
45 0.237450567
50 0.22523034
};
\addplot [semithick, darkblue, dashed, mark=*, mark size=1, mark options={solid}]
table {%
10 0.269871788
15 0.160029781
20 0.105513731
25 0.075790195
30 0.057724052
35 0.045684087
40 0.03706195
45 0.0307127
50 0.02588665
};
\addplot [semithick, gray, dashed, mark=square*, mark size=1, mark options={solid}]
table {%
10 0.707007585
15 0.530140861
20 0.381037381
25 0.279523066
30 0.211543445
35 0.165381606
40 0.132546618
45 0.108041709
50 0.089874093
};
\addplot [semithick, lightskyblue, dashed, mark=triangle*, mark size=1, mark options={solid}]
table {%
10 0.679938452
15 0.582100886
20 0.467624396
25 0.359890593
30 0.273716302
35 0.210166652
40 0.164274353
45 0.130783691
50 0.106644331
};
\end{axis}

\end{tikzpicture}
    }
    \qquad
    \subfloat[Avg. energy, correlated\label{fig:E_correlated}]{
\begin{tikzpicture}

\begin{axis}[
scale only axis,
width=\four,
height=\four,
xlabel={\(\displaystyle \lambda_a = \lambda_q\) [packets/s]},
xmajorgrids,
xmin=8, xmax=52,
ylabel={\(\displaystyle \bar{E}\) [mJ]},
ymajorgrids,
ymin=0.12, ymax=1.08,
]
\addplot [draw=darkblue, fill=darkblue, forget plot, mark=*, mark size=1, only marks]
table{%
x  y
10 0.208652972
15 0.269686639
20 0.323851758
25 0.37394654
30 0.417895562
35 0.458679963
40 0.496160103
45 0.530009849
50 0.561174586

};
\addplot [draw=gray, fill=gray, forget plot, mark=square*, mark size=1, only marks]
table{%
x  y
10 0.400335629
15 0.465602237
20 0.523605305
25 0.574232715
30 0.62086247
35 0.662610394
40 0.69990453
45 0.735192572
50 0.765987783

};
\addplot [draw=lightskyblue, fill=lightskyblue, forget plot, mark=triangle*, mark size=1, only marks]
table{%
x  y
10 0.66028142
15 0.731729425
20 0.787968666
25 0.83832608
30 0.883433736
35 0.922774277
40 0.957591898
45 0.987421596
50 1.011752223
};
\addplot [semithick, darkblue, forget plot, dashed, dash pattern=on 5pt off 3pt on 1pt off 3pt,] 
table {%
10 0.208652972
15 0.269686639
20 0.323851758
25 0.37394654
30 0.417895562
35 0.458679963
40 0.496160103
45 0.530009849
50 0.561174586

};
\addplot [semithick, gray, forget plot, dashed, dash pattern=on 5pt off 3pt on 1pt off 3pt,] 
table {%
10 0.400335629
15 0.465602237
20 0.523605305
25 0.574232715
30 0.62086247
35 0.662610394
40 0.69990453
45 0.735192572
50 0.765987783
};
\addplot [semithick, lightskyblue, forget plot, dashed, dash pattern=on 5pt off 3pt on 1pt off 3pt,] 
table {%
10 0.66028142
15 0.731729425
20 0.787968666
25 0.83832608
30 0.883433736
35 0.922774277
40 0.957591898
45 0.987421596
50 1.011752223
};

\end{axis}

\end{tikzpicture}
    }
    \caption{\revise{Comparison of analytical derivation (lines) and numerical simulations (markers) of average success probability and energy consumption metrics vs. average incoming anomalies $\lambda_a$ and queries $\lambda_q$ for different $Q$ values and with or without alarm correlations.}}
    \label{fig:analytical}
\end{figure}

Figs.~\ref{fig:analytical_alarm_query} \revise{and~\ref{fig:pq_pa_correlated}} show the average success probability of queries, $\bar{p}_q$, and alarms, $\bar{p}_a$, versus the average  arrival rate $\lambda_q$ and  $\lambda_a$ \revise{for uncorrelated and correlated traffic}, respectively.
We set $\lambda_q = \lambda_a$ to observe how the system behaves when the query and alarm traffic loads are equal on average. The results compare our theoretical analysis (lines) and computer simulation (markers) considering $Q \in \{1, 4, 7\}$. A perfect correspondence between the theoretical analysis and the numerical simulations is demonstrated, thereby validating our results. Also, as $\lambda_a =\lambda_q$ increases, both $\bar{p}_a$ and $\bar{p}_q$ decrease. 
\revise{This is due to the higher collision rate among alarmed nodes during the push sub-frame, which directly reduces $\bar{p}_a$. Consequently, the accumulation of retransmitted alarms, combined with the higher query arrival rate, increases the contention for available slots, thereby lowering $\bar{p}_q$. The results also highlight a fundamental trade-off: increasing $Q$ improves $\bar{p}_q$ at the expense of $\bar{p}_a$, and vice versa. While correlated alarms yield similar trends to the uncorrelated case, a slight improvement in $\bar{p}_a$ is observed for $Q=7$. This is attributed to the higher probability of ``empty" frames under correlated traffic; the absence of active alarms in these frames eliminates collisions ($p_a(t)=1$), thereby reducing the overall retransmission burden and indirectly increasing $\bar{p}_q$.}

Figs.~\ref{fig:E_vs_lambda} and~\ref{fig:E_correlated} illustrate the average energy consumed by the system, $\bar{E}$, as a function of $\lambda_q$ and $\lambda_a$, \revise{for uncorrelated and correlated traffic, respectively}. The parameter settings remain the same as in the previous graph, and a perfect correspondence between the theoretical analysis and the numerical simulation is still observable. Furthermore, as $\lambda_a =\lambda_q$ increases, the average energy consumed by the system also increases. This is due to a higher input load, which leads to a greater number of transmission attempts 
and, consequently, higher power consumption. Remark that the configuration consuming the lowest amount of energy for a given input traffic is the one with $Q = 1$. This is because fewer nodes receive the \gls{wus}, leading to a lower $E_1(t)$ compared to $Q = \{4,7\}$.  Moreover, thanks to a greater number of available push slots, the number of retransmissions is lower, which results in a lower $E_2(t)$. Additionally, $E_3(t)$ is reduced since the \gls{wur} remains active for a shorter duration.
\revise{However, when $Q=1$, the query success probability, $\bar{p}_q$, degrades significantly as the frame is dominated by alarm traffic, leaving insufficient resources for query handling. Finally, correlated traffic exhibits a marginally higher $\bar{E}$ compared to the uncorrelated scenario, as bursts of simultaneous alarms intensify retransmissions and per-frame energy expenditure.}

The previous findings highlight the importance of carefully selecting the parameter $Q$ to account for the performance of both types of traffic, as well as the system power consumption.


\subsection{System-level performance characterization}

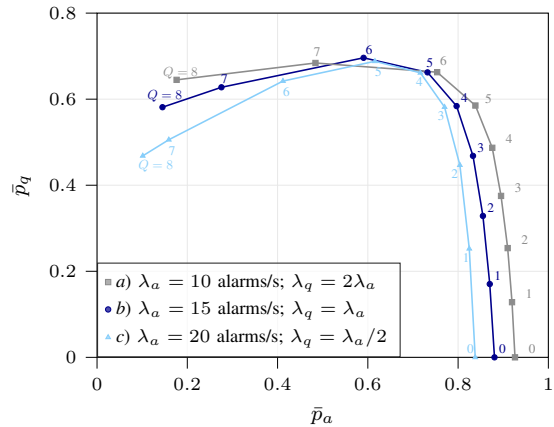
\begin{figure}[t!]
\centering
\begin{tikzpicture}

\begin{axis}[
height=4cm,
width=5.5cm,
scale only axis,
legend cell align={left},
legend style={at={(0, 0)}, anchor=south west, /tikz/every even column/.append style={column sep=0.2cm}},
tick align=outside,
tick pos=left,
xlabel={\(\displaystyle \bar{p}_a\)},
xmajorgrids,
xmin=0, xmax=1.0,
ylabel={\(\displaystyle \bar{p}_q\)},
ymajorgrids,
ymin=0, ymax=0.8,
]
\addplot [draw=gray, fill=gray, mark=square*, only marks, mark size=1pt]
table{%
x  y
0.92607227 0
0.919486326 0.128216947
0.909850154 0.253881729
0.895161246 0.37518072
0.875695038 0.487013792
0.838336666 0.585225053
0.753618534 0.662717853
0.483969734 0.684095376
0.176284865 0.644916659
};
\addlegendentry{\emph{a}) $\lambda_a = 10$ alarms/s; $\lambda_q = 2 \lambda_a$}

\addplot [draw=darkblue, fill=darkblue, mark=*, only marks, mark size=1pt]
table{%
x  y
0.880633291 0
0.870165463 0.170369497
0.855007322 0.328622562
0.833002691 0.468227055

0.796645485 0.58393221
0.732271766 0.662284487
0.590585154 0.696107032
0.275649132 0.627631324
0.145004039 0.581327069
};
\addlegendentry{\emph{b}) $\lambda_a = 15$ alarms/s; $\lambda_q = \lambda_a$}

\addplot [draw=lightskyblue, fill=lightskyblue, mark=triangle*, only marks, mark size=1pt]
table{%
x  y
0.837949583 0
0.824771702 0.25244102
0.8039651 0.446685365
0.770025768 0.58158358
0.716466868 0.661238959
0.615811244 0.688030338
0.41196889 0.642101368
0.159023834 0.505724076
0.101357746 0.468274431
};
\addlegendentry{\emph{c}) $\lambda_a = 20$ alarms/s; $\lambda_q = \lambda_a / 2$}

\addplot [semithick, darkblue, forget plot]
table {%
0.880633291 0
0.870165463 0.170369497
0.855007322 0.328622562
0.833002691 0.468227055
0.796645485 0.58393221
0.732271766 0.662284487
0.590585154 0.696107032
0.275649132 0.627631324
0.145004039 0.581327069
};
\addplot [semithick, gray, forget plot]
table {%
0.92607227 0
0.919486326 0.128216947
0.909850154 0.253881729
0.895161246 0.37518072
0.875695038 0.487013792
0.838336666 0.585225053
0.753618534 0.662717853
0.483969734 0.684095376
0.176284865 0.644916659
};
\addplot [semithick, lightskyblue, forget plot]
table {%
0.837949583 0
0.824771702 0.25244102
0.8039651 0.446685365
0.770025768 0.58158358
0.716466868 0.661238959
0.615811244 0.688030338
0.41196889 0.642101368
0.159023834 0.505724076
0.101357746 0.468274431
};

\draw (axis cs:0.880633291,0.01) node[
  scale=0.5,
  anchor=base west,
  text=darkblue,
  rotate=0.0
]{0};
\draw (axis cs:0.94607227,0.01) node[
  scale=0.5,
  anchor=base west,
  text=gray,
  rotate=0.0
]{0};
\draw (axis cs:0.81,0.01) node[
  scale=0.5,
  anchor=base west,
  text=lightskyblue,
  rotate=0.0
]{0};
\draw (axis cs:0.870165463,0.180369497) node[
  scale=0.5,
  anchor=base west,
  text=darkblue,
  rotate=0.0
]{1};
\draw (axis cs:0.939486326,0.138216947) node[
  scale=0.5,
  anchor=base west,
  text=gray,
  rotate=0.0
]{1};
\draw (axis cs:0.8,0.22244102) node[
  scale=0.5,
  anchor=base west,
  text=lightskyblue,
  rotate=0.0
]{1};
\draw (axis cs:0.855007322,0.338622562) node[
  scale=0.5,
  anchor=base west,
  text=darkblue,
  rotate=0.0
]{2};
\draw (axis cs:0.929850154,0.263881729) node[
  scale=0.5,
  anchor=base west,
  text=gray,
  rotate=0.0
]{2};
\draw (axis cs:0.775,0.416685365) node[
  scale=0.5,
  anchor=base west,
  text=lightskyblue,
  rotate=0.0
]{2};
\draw (axis cs:0.833002691,0.478227055) node[
  scale=0.5,
  anchor=base west,
  text=darkblue,
  rotate=0.0
]{3};
\draw (axis cs:0.915161246,0.38518072) node[
  scale=0.5,
  anchor=base west,
  text=gray,
  rotate=0.0
]{3};
\draw (axis cs:0.745,0.55158358) node[
  scale=0.5,
  anchor=base west,
  text=lightskyblue,
  rotate=0.0
]{3};
\draw (axis cs:0.796645485,0.59393221) node[
  scale=0.5,
  anchor=base west,
  text=darkblue,
  rotate=0.0
]{4};
\draw (axis cs:0.895695038,0.497013792) node[
  scale=0.5,
  anchor=base west,
  text=gray,
  rotate=0.0
]{4};
\draw (axis cs:0.696466868,0.631238959) node[
  scale=0.5,
  anchor=base west,
  text=lightskyblue,
  rotate=0.0
]{4};
\draw (axis cs:0.605811244,0.653030338) node[
  scale=0.5,
  anchor=base west,
  text=lightskyblue,
  rotate=0.0
]{5};
\draw (axis cs:0.85,0.59) node[
  scale=0.5,
  anchor=base west,
  text=gray,
  rotate=0.0
]{5};
\draw (axis cs:0.72,0.675) node[
  scale=0.5,
  anchor=base west,
  text=darkblue,
  rotate=0.0
]{5};
\draw (axis cs:0.40196889,0.607101368) node[
  scale=0.5,
  anchor=base west,
  text=lightskyblue,
  rotate=0.0
]{6};
\draw (axis cs:0.75,0.68) node[
  scale=0.5,
  anchor=base west,
  text=gray,
  rotate=0.0
]{6};
\draw (axis cs:0.585,0.705) node[
  scale=0.5,
  anchor=base west,
  text=darkblue,
  rotate=0.0
]{6};
\draw (axis cs:0.149023834,0.470724076) node[
  scale=0.5,
  anchor=base west,
  text=lightskyblue,
  rotate=0.0
]{7};
\draw (axis cs:0.47,0.7) node[
  scale=0.5,
  anchor=base west,
  text=gray,
  rotate=0.0
]{7};
\draw (axis cs:0.265,0.637) node[
  scale=0.5,
  anchor=base west,
  text=darkblue,
  rotate=0.0
]{7};
\draw (axis cs:0.15,0.60) node[
  scale=0.5,
  anchor=base,
  text=darkblue,
  rotate=0.0
]{$Q=8$};
\draw (axis cs:0.136284865,0.654916659) node[
  scale=0.5,
  anchor=base west,
  text=gray,
  rotate=0.0
]{$Q=8$};
\draw (axis cs:0.081357746,0.438274431) node[
  scale=0.5,
  anchor=base west,
  text=lightskyblue,
  rotate=0.0
]{$Q=8$};
\end{axis}

\end{tikzpicture}
    \caption{Success probability of queries, $\bar{p}_q$ as a function of the success probability of alarms, $\bar{p}_a$, when $Q\in[0, 8]$, for different $\lambda_a$ and $\lambda_q$ values.} 
    \label{fig:pq_vs_pa}
\end{figure}

Fig.~\ref{fig:pq_vs_pa} shows the success probability of queries, $\bar{p}_q$, as a function of the success probability of alarms, $\bar{p}_a$, when $Q\in[0,8]$, \revise{considering only uncorrelated traffic, since the results for the correlated case are similar}. Three different case studies are analyzed, all characterized by the same total input traffic load, $\lambda = \lambda_a + \lambda_q = 30$ [packets/s]: \emph{a}) $\lambda_a = 10$ and $\lambda_q = 2 \lambda_a$; \emph{b}) $\lambda_a = \lambda_q = 15$;  \emph{c}) $\lambda_a = 20$ and $\lambda_q = \lambda_a / 2$.  
Each point in the figure represents the values of $\bar{p}_a$ and $\bar{p}_q$ obtained for a specific $Q$, corresponding to a particular configuration of pull and push sub-frames durations.
When $Q = 0$, no pull slots are present, meaning the entire frame is reserved for push-based communication and alarm handling. In this case, $\bar{p}_q$ is zero regardless of the input traffic because no query can be served. 
Conversely, with $Q = 8$, no push slots are present, but alarmed nodes can still report their anomalies if they receive a \gls{wus} -- cf. Assumption~\ref{assu:alarm-tx}. 
Therefore, the alarmed nodes contend for the resources -- the $Q$ pull slots -- with the non-alarmed ones, explaining why $\bar{p}_a$ is not zero with $Q=8$. 
This also explains why $\bar{p}_q$ degrades when $Q=8$ compared to lower values of $Q$: when no push slots are considered, all the alarm traffic is multiplexed in the pull slots, hindering the performance of serving queries. This is also corroborated by the fact that the highest degradation happens for the case study \emph{c}), which has the highest alarm traffic load.
Also, $\bar{p}_q$ increases gradually as $Q$ decreases until reaching an optimal point; beyond that, performance degrades. The optimal values for queries are $Q=\{7, 6, 5\}$ for the three use cases, respectively. According to the previous discussion, the higher the alarm traffic load, the lower is the value of $Q$ achieving the maximum $\bar{p}_q$.
On the other hand, the performance in reporting anomalies increases monotonically with the decrease of $Q$. Indeed, when $Q$ decreases, $P$ increases, resulting in a higher push success probability $p_p(t)$ -- see eq.~\eqref{eq:pp(t)}.  Moreover, the achievable $\bar{p}_a$ is higher in the case \emph{a}) compared to the cases \emph{b}) and \emph{c}) as lower alarm traffic load leads to fewer collisions within the push sub-frame. 

\begin{figure}[t!]
\centering
\subfloat{\centering\begin{tikzpicture}
\begin{axis}[
    width=0cm,
    height=0cm,
    axis line style={draw=none},
    tick style={draw=none},
    at={(0,0)},
    scale only axis,
    xmin=0,
    xmax=1,
    xtick={},
    ymin=0,
    ymax=1,
    ytick={},
    legend cell align={left},
    legend style={at={(0.5, 1,05)}, anchor=south, draw=none, fill=none, /tikz/every even column/.append style={column sep=5pt}},
    legend image code/.code={
            \draw [#1] (0cm,-0.1cm) rectangle (0.15cm,4pt); },
    legend columns = 3,
    ]

\addlegendimage{draw=white,fill=gray}
\addlegendentry{\emph{a}) $\lambda_a =$ 10; $\lambda_q = 2 \lambda_a$}
\addlegendimage{draw=darkblue,fill=darkblue}
\addlegendentry{\emph{b}) $\lambda_a = \lambda_q = 15$}
\addlegendimage{draw=lightskyblue,fill=lightskyblue}
\addlegendentry{\emph{c}) $\lambda_a =$ 20; $\lambda_q = \lambda_a / 2$}
\addlegendimage{draw=darkblue,fill=darkblue}
\end{axis}
\end{tikzpicture}}\\
\setcounter{subfigure}{0}
\subfloat[Uncorrelated alarms. \label{fig:ps_vs_Q_uncorrelated}]{\begin{tikzpicture}

\begin{axis}[
    height=3.3cm,
    width=\twowidth,
    scale only axis,
    xlabel={$Q$},
    xmin=-0.5, xmax=8.5,
    xtick={0,1,...,8},
    xlabel shift=-3.5pt,
    ylabel={$\bar{p}_s$},
    ymin=0, ymax=0.743,
    ylabel shift=-3pt,
    ybar=0pt,
    bar width=2pt,
    xtick distance=3.5mm,
    tick align=inside,
    ]

    \addplot[style={gray,fill={gray}}] 
    table[x=Q, y=ps, col sep=comma] 
    {tikz/csv/uncorrelated_a10_q20.csv};
    
    \addplot[style={darkblue,fill={darkblue}}] 
    table[x=Q, y=ps, col sep=comma] 
    {tikz/csv/uncorrelated_a15_q15.csv};
    
    \addplot[style={lightskyblue,fill={lightskyblue}}] 
    table[x=Q, y=ps, col sep=comma] 
    {tikz/csv/uncorrelated_a20_q10.csv};

\end{axis}
\end{tikzpicture}}
\qquad
\subfloat[Correlated alarms. \label{fig:ps_vs_Q_correlated}]{\begin{tikzpicture}
\begin{axis}[
    height=3.3cm,
    width=\twowidth,
    scale only axis,
    xlabel={$Q$},
    xmin=-0.5, xmax=8.5,
    xtick={0,1,...,8},
    xlabel shift=-3.5pt,
    ymin=0, ymax=0.743,
    ybar=0pt,
    bar width=2pt,
    xtick distance=3.5mm,
    tick align=inside,
    ]
    
    \addplot[style={gray,fill={gray}}] 
    table[x=Q, y=ps, col sep=comma] 
    {tikz/csv/correlated_a10_q20.csv};
    
    \addplot[style={darkblue,fill={darkblue}}] 
    table[x=Q, y=ps, col sep=comma] 
    {tikz/csv/correlated_a15_q15.csv};
    
    \addplot[style={lightskyblue,fill={lightskyblue}}] 
    table[x=Q, y=ps, col sep=comma] 
    {tikz/csv/correlated_a20_q10.csv};

\end{axis}
\end{tikzpicture}}
    \caption{\revise{Success probability trade-off, $\bar{p}_s$ as a function of $Q$ for different $\lambda_a$ [alarms/s] and $\lambda_q$ [queries/s].}}
    \label{fig:ps_vs_Q}
\end{figure}
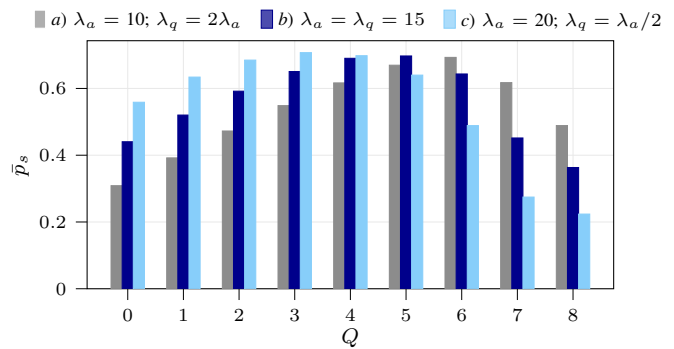

\revise{Figs.~\ref{fig:ps_vs_Q_uncorrelated} and ~\ref{fig:ps_vs_Q_correlated} show} the behavior of the success probability trade-off, $\bar{p}_s$, as a function of $Q$, for the three case studies defined in Fig.~\ref{fig:pq_vs_pa} \revise{and for uncorrelated and correlated traffic, respectively}. \revise{Results in Fig. \ref{fig:ps_vs_Q_uncorrelated}} show that the highest $p_s$ is obtained with $Q = \{6, 5, 3\}$, for cases \emph{a}), \emph{b}), and \emph{c}), respectively. As expected, when queries have a higher traffic load than alarms ($\lambda_q > \lambda_a$), the maximum value of $\bar{p}_s$ is achieved by a higher $Q$, meaning that more slots are dedicated to pull-based communication; if the traffic of alarms is higher ($\lambda_q < \lambda_a$), lower values of $Q$ are preferable, also to lower the probability of alarmed nodes receiving the \gls{wus}. \revise{In the correlated alarm case of Fig. \ref{fig:ps_vs_Q_correlated}, the optimal trade-off slightly shifts and the maximum $\bar{p}_s$ is obtained for $Q=\{5,4,4\}$ in cases \emph{a}), \emph{b}), and \emph{c}), respectively. Although the overall trends remain similar to the uncorrelated scenario, correlation increases variability in the number of alarms per frame, creating both empty frames and alarm-dense frames. As a result, performance depends on both the input traffic load and the specific distribution of alarms. When $\lambda_a > \lambda_q$, alarm-dense frames have a strong impact on performance, reducing both $\bar{p}_q$ and the overall success probability $\bar{p}_s$, particularly for small $Q$. Conversely, when $\lambda_q > \lambda_a$, empty or lightly loaded frames dominate, leading to higher $\bar{p}_s$ compared to the uncorrelated case.}

\begin{figure}[t!]
\centering
\subfloat{\centering\begin{tikzpicture}
\begin{axis}[
    width=0cm,
    height=0cm,
    axis line style={draw=none},
    tick style={draw=none},
    at={(0,0)},
    scale only axis,
    xmin=0,
    xmax=1,
    xtick={},
    ymin=0,
    ymax=1,
    ytick={},
    legend cell align={left},
    legend style={at={(0.5, 1,05)}, anchor=south, draw=none, fill=none, /tikz/every even column/.append style={column sep=2pt}},
    legend columns = 3,
    ]

\addlegendimage{gray, only marks, mark=square*, mark options={fill=gray}, mark size=2}
\addlegendentry{\emph{a}) $\lambda_a =$ 10; $\lambda_q = 2 \lambda_a$}
    
\addlegendimage{darkblue, only marks, mark=*, mark options={fill=darkblue}, mark size=2}
\addlegendentry{\emph{b}) $\lambda_a = \lambda_q = 15$}

\addlegendimage{lightskyblue, only marks, mark=triangle*, mark options={fill=lightskyblue}, mark size=2,}
\addlegendentry{\emph{c}) $\lambda_a =$ 20; $\lambda_q = \lambda_a / 2$}

\end{axis}
\end{tikzpicture}}\\
\vspace{-3.5mm}
\setcounter{subfigure}{0}
\subfloat[Uncorrelated alarms.]{ \begin{tikzpicture}
\begin{axis}[
    width=\twowidth,
    height=\twowidth,
    scale only axis,
    xlabel={$\bar{E}$ [mJ]},
    xmin=0, xmax=0.88,
    ylabel={$\bar{p}_s$},
    ymin=0.2, ymax=0.75,
    nodes near coords,
    point meta=explicit symbolic,
    every node near coord/.append style={
        font=\tiny,
        anchor=south,
    },
]

\addplot[draw=gray, fill=gray, mark=square*, mark size=1, only marks, nodes near coords style={color=gray, anchor=north, inner sep=0.9pt}] 
    table[x=E, y=ps, meta=label, col sep=comma] {tikz/csv/uncorrelated_a10_q20.csv};
\addplot[draw=gray, semithick] 
table[x=E, y=ps, col sep=comma] {tikz/csv/uncorrelated_a10_q20.csv};

\addplot[draw=lightskyblue, fill=lightskyblue, mark=triangle*, mark size=1, only marks, nodes near coords style={color=lightskyblue, anchor=south, inner sep=1pt, xshift=0pt}] 
    table[x=E, y=ps, meta=label, col sep=comma] {tikz/csv/uncorrelated_a20_q10.csv};

\addplot[draw=lightskyblue, semithick] table[x=E, y=ps, col sep=comma] {tikz/csv/uncorrelated_a20_q10.csv};

\addplot[draw=darkblue, fill=darkblue, mark=*, mark size=1, only marks, nodes near coords style={color=darkblue, anchor=west, inner sep=0.9pt, xshift=0pt}] 
    table[x=E, y=ps, meta=label, col sep=comma] {tikz/csv/uncorrelated_a15_q15.csv};

\addplot[draw=darkblue, semithick] table[x=E, y=ps, col sep=comma] {tikz/csv/uncorrelated_a15_q15.csv};

\end{axis}
\end{tikzpicture}\label{fig:E_vs_ps_uncorrelated}} \qquad
\subfloat[Correlated alarms.]{\begin{tikzpicture}
\begin{axis}[
    width=\twowidth,
    height=\twowidth,
    scale only axis,
    xlabel={$\bar{E}$ [mJ]},
    xmin=0, xmax=0.88,
    ylabel={$\bar{p}_s$},
    ymin=0.2, ymax=0.75,
    nodes near coords,
    point meta=explicit symbolic,
    every node near coord/.append style={
        font=\tiny,
        anchor=south,
    },
]

\addplot[draw=gray, fill=gray, mark=square*, mark size=1, only marks, nodes near coords style={color=gray, anchor=south}] 
    table[x=E, y=ps, meta=label, col sep=comma] {tikz/csv/correlated_a10_q20.csv};
\addplot[draw=gray, semithick] 
table[x=E, y=ps, col sep=comma] {tikz/csv/correlated_a10_q20.csv};

\addplot[draw=darkblue, fill=darkblue, mark=*, mark size=1, only marks, nodes near coords style={color=darkblue, anchor=south, inner sep=1pt, xshift=2pt}] 
    table[x=E, y=ps, meta=label, col sep=comma] {tikz/csv/correlated_a15_q15.csv};

\addplot[draw=darkblue, semithick] table[x=E, y=ps, col sep=comma] {tikz/csv/correlated_a15_q15.csv};

\addplot[draw=lightskyblue, fill=lightskyblue, mark=triangle*, mark size=1, only marks, nodes near coords style={color=lightskyblue, anchor=north, inner sep=0.8pt}] 
    table[x=E, y=ps, meta=label, col sep=comma] {tikz/csv/correlated_a20_q10.csv};

\addplot[draw=lightskyblue, semithick] table[x=E, y=ps, col sep=comma] {tikz/csv/correlated_a20_q10.csv};

\end{axis}
\end{tikzpicture}\label{fig:E_vs_ps_correlated}}
\caption{\revise{Success probability--enegy trade-off, $\bar{p}_s$ vs. $\bar{E}$, for different $\lambda_a$, $\lambda_q$.}}
    \label{fig:E_vs_ps}
    \vspace{-5mm}
\end{figure}

\revise{Figs.~\ref{fig:E_vs_ps_uncorrelated} and ~\ref{fig:E_vs_ps_correlated} illustrate} the success probability trade-off, $\bar{p}_s$, as a function of the average energy consumed by the system, $\bar{E}$, for $Q\in[0,8]$ \revise{under uncorrelated and correlated traffic}, considering the three case studies analyzed previously.
\revise{In both figures, as} $Q$ increases from 0 to 7, the energy consumption of the system also increases. However, when $Q = 8$\revise{--meaning that no slots are reserved for the push sub-frame--}energy consumption decreases. This occurs because the alarmed nodes \revise{transmit only if triggered by a \gls{wus}, causing the contribution of $E_2(t)$ in~\eqref{eq: E2} to drop to zero.}  
Conversely, the lowest energy consumption is observed for $Q = 0$, where the system operates exclusively in push mode, \revise{and it is} unable to serve any queries.  
\revise{In this case, only alarmed nodes} contribute to the energy expenditure. 
\revise{Although the qualitative behavior of Figs.~\ref{fig:E_vs_ps_uncorrelated} and~\ref{fig:E_vs_ps_correlated} is similar, quantitative differences arise due to traffic correlation and the relative values of $\lambda_a$ and $\lambda_q$. In presence of correlated alarms, for case study c), the energy consumption for small $Q$ is higher than in case a), whereas the opposite trend is observed in the uncorrelated scenario. This difference stems from the different temporal distribution of alarms, which modifies the contribution of each sub-frame to the overall energy consumption.}
\revise{In Fig. ~\ref{fig:E_vs_ps_uncorrelated},} the values of $Q$ maximizing $\bar{p}_s$ correspond to different \textcolor{black}{values} of energy spent: for \emph{a}) $Q = 6$, $\bar{p}_s \approx 0.69$ and $E \approx 471\, \mu$J; for \emph{b}), $Q = 5$, $\bar{p}_s \approx 0.70$ and $\bar{E} \approx 466\,\mu$J; in \emph{c}) $Q = 3$, $p_s \approx 0.71$ and $E \approx 394\,\mu$J. \revise{In contrast, in Fig.~\ref{fig:E_vs_ps_correlated}, the optimal values shift to $Q \in \{5,4,4\}$. However, in all the cases,} the maximum $\bar{p}_s$ is achieved with an intermediate energy consumption. 
Therefore, the optimal values of $Q$ for $\bar{p}_s$ balance energy consumption while ensuring system functionality.

\begin{figure}[t!]
\centering
\subfloat{\begin{tikzpicture}
\begin{axis}[
    width=0cm,
    height=0cm,
    axis line style={draw=none},
    tick style={draw=none},
    at={(0,0)},
    scale only axis,
    xmin=0,
    xmax=1,
    xtick={},
    ymin=0,
    ymax=1,
    ytick={},
    legend cell align={left},
    legend style={at={(0.5, 1,05)}, anchor=south, draw=none, fill=none, /tikz/every even column/.append style={column sep=5pt}},
    legend image code/.code={        
            \draw[#1] (0cm,-0.1cm) rectangle (0.15cm,4pt);
    },
    legend columns=4,
    ]
MR
\addlegendimage{draw=darkblue,fill=darkblue}
\addlegendentry{WuR}
\addlegendimage{draw=orange,fill=orange}
\addlegendentry{MR}
\addlegendimage{black, dashed, semithick, line legend}
\addlegendentry{RR}
\addlegendimage{black, dotted, semithick, line legend}
\addlegendentry{FSA}

\end{axis}
\end{tikzpicture}}\\
\vspace{-3mm}
\setcounter{subfigure}{0}
\subfloat[Avg. energy consumption\label{fig:wur_vs_benchmarks}]{\begin{tikzpicture}
\begin{axis}[
    height=\twowidth,
    width=\twowidth,
    scale only axis,
    xlabel={$P$},
    xmin=3, xmax=37,
    xtick={5,10,15,...,35},
    xlabel shift=-3.5pt,
    ylabel={$\bar{E}$ [mJ]},
    ymin=0.1, ymax=0.9,
    xtick distance=3.5mm,
    tick align=inside,
    ]
    
    \addplot[style={darkblue,fill={darkblue}}, bar width=2pt, ybar=0pt, bar shift=-1pt] 
    table[x=P, y=E_wur, col sep=comma] 
    {tikz/csv/benchmarks.csv};
    
    \addplot[style={orange,fill={orange}}, bar width=2pt, ybar=0pt, bar shift=1pt] 
    table[x=P, y=E_mr, col sep=comma] 
    {tikz/csv/benchmarks.csv};
    


\addplot [black, dashed, semithick] 
table{
0 0.642
40 0.642
};
\addplot [black, dotted, semithick] 
table{
0 0.66
40 0.66
};

\end{axis}
\end{tikzpicture}} \qquad
\subfloat[Avg. energy efficiency\label{fig:E_G_vs_P}]{\begin{tikzpicture}
\begin{axis}[
    height=\twowidth,
    width=\twowidth,
    scale only axis,
    xlabel={$P$},
    xmin=3, xmax=37,
    xtick={5,10,15,...,35},
    xlabel shift=-3.5pt,
    ylabel={$\bar{E} / \bar{S}$ [mJ]},
    ymin=0.01, ymax=0.11,
    ytick={0.02, 0.04, 0.06, 0.08, 0.1},
    yticklabels={0.02, 0.04, 0.06, 0.08, 0.10},
    xtick distance=3.5mm,
    tick align=inside,
    ]
    
    \addplot[style={darkblue,fill={darkblue}}, bar width=2pt, ybar=0pt, bar shift=-1pt] 
    table[x=P, y=ES_wur, col sep=comma] 
    {tikz/csv/benchmarks_ES.csv};
    
    \addplot[style={orange,fill={orange}}, bar width=2pt, ybar=0pt, bar shift=1pt] 
    table[x=P, y=ES_mr, col sep=comma] 
    {tikz/csv/benchmarks_ES.csv};
    
\addplot [black, dashed, semithick] 
table{
0 0.059847485
40 0.059847485
};
\addplot [black, dotted, semithick] 
table{
0 0.074188977
40 0.074188977
};

\end{axis}
\end{tikzpicture}}
\caption{\revise{Energy performance comparison vs $P$, $\lambda_a =\lambda_q = 15$ packets/s.}}    
\label{fig:E_benchmarks}
\vspace{-4mm}
\end{figure}

\subsection{Comparison with benchmarks}
\label{sub:wur_vs_benchmarks}

We compare the proposed \gls{wur}-based solution with three baseline schemes not employing \gls{wur} technology, named Main Radio (MR), Round Robin (RR), and \gls{fsa}. \revise{The comparison is performed only for uncorrelated alarms, since the trends and qualitative behavior are the same for correlated traffic, making the results representative of both cases.}

\emph{MR} \revise{mimics the proposed solution but} all nodes keep their main radio continuously active.
\revise{At the frame start, the \gls{bs} collects all queries from the cloud; then, it broadcasts a \emph{single} control message lasting for $k_s$ slots, encoding 
the scheduled nodes' \glspl{id} and their assigned pull sub-frames' slots, while alarmed nodes attempt access in the push sub-frame.} 
The maximum number of queries that can be served is $Q' = F - P - k_s - k_c$, \revise{\textcolor{black}{where} $Q' \ge Q$ of the \gls{wus}-based solution, at the cost of higher energy consumption due to continuous main-radio operations.}
\revise{In \emph{RR},} nodes are scheduled sequentially according to their \glspl{id}; if the frame cannot accommodate all nodes, the remaining ones are scheduled in subsequent frames. \revise{The first $k_s$ slots of each frame carry a \gls{bs}' control signal encoding the scheduling and the received queries.} Nodes transmit only if they have been queried or are in alarm \revise{to avoid unnecessary transmissions.} 
\revise{\emph{\gls{fsa}} employs the grant-free \gls{fsa} protocol throughout the frame.} The first $k_s$ slots carry a \revise{\gls{bs}' control signal} listing the nodes for which queries have been received. \revise{Queried and alarmed n}odes then attempt \revise{access, randomly selecting a single slot.} 
Collisions may occur, with alarmed nodes retrying in subsequent frames. 
\revise{The maximum number of queries that can be served for RR and FSA is $Q'' = F - k_s$, where $Q'' > Q' > Q$.} 



%

In our evaluation, we consider the best case \revise{benchmarks} in which the control signaling fits within a single slot, \revise{i.e., $k_s = 1$. \revise{This maximizes the RR and \gls{fsa} throughput \textcolor{black}{since} the number of \textcolor{black}{\glspl{ue} is} equal to the number of remaining slots ($F - k_s = N = 40$).} Then, the average energy is evaluated over $T_O$ frames, considering that the energy consumption consists of three main components: $1)$ a constant energy term, since all nodes listen during the initial $k_s$ slots to determine whether they are scheduled or queried; $2)$ the transmission energy of the queried nodes during the $Q'$ or $Q''$ slots; and $3)$ the energy consumed by the remaining alarmed nodes, which \textit{a)} under MR attempt transmission in the push sub-frame, \textit{b)} under RR transmit in their assigned slots, or \textit{c)} under \gls{fsa} randomly select one of the slots of the frame for transmission.}


\revise{Fig.~\ref{fig:wur_vs_benchmarks} compares the average energy consumption of the proposed \gls{wur} solution with MR, RR, and \gls{fsa}. For \gls{wur} and MR, the frame is divided into pull and push sub-frames, so energy varies with $P$, the number of push slots. In contrast, RR and \gls{fsa} use a single frame serving up to $Q'' = 40$ queries, resulting in constant but higher energy consumption. As $P$ increases, energy decreases for both \gls{wur} and MR, since more push slots improve alarm success and reduce retransmissions. MR consistently consumes more energy than \gls{wur}, because all nodes must keep the main radio active to receive scheduling signals, unlike the low-power \gls{wus} reception. RR and \gls{fsa} also require more energy, except for \gls{wur} at $P=5$, where low alarm success leads to additional retransmissions. Their higher energy mainly stems from keeping the main radio active for multiple slots. Note that, since $N = 40$, RR can schedule all nodes within a single frame, representing an optimal scenario for this scheme. Even under this best-case assumption for all benchmarks, where control signaling fits in a single slot, the \gls{wur} solution still outperforms them, demonstrating its efficiency.} 

To perform a fair comparison of energy expenditure, we calculate the average energy spent by the system to successfully \revise{serve either a query or an alarm, i.e., the average \emph{energy efficiency}}. To do so, we denote as $\bar{S}$ the average number of packets successfully transmitted within a frame. \revise{In particular, $\bar{S} = \mu_q \bar{p}_q + \bar{\mu}_a \bar{p}_a$}
\revise{is a function of the average number of queries per frame $\mu_q$, and the average number of alarm packets in a frame,  $\bar{\mu}_a = \frac{1}{T_O} \sum_{t=0}^{T_O-1} \E{n_a(t)}$, 
computed using~\eqref{eq:na(t+1)}.}

\revise{Accordingly, Fig.~\ref{fig:E_G_vs_P} reports the average energy \revise{efficiency}, $\bar{E}/\bar{S}$, as a function of $P$. In addition to the proposed \gls{wur} solution, the figure also includes the results for MR, RR and \gls{fsa}. As in Fig.~\ref{fig:wur_vs_benchmarks}, the performance of RR and \gls{fsa} remains constant with $P$, whereas MR and \gls{wur} vary with the push/query slot allocation.
Overall, the proposed \gls{wur} solution achieves the lowest $\bar{E}/\bar{S}$ for most of the considered range of $P$. In particular, for $P > 5$, it consistently outperforms MR, RR and \gls{fsa}, confirming its superior energy efficiency per successfully served packet.
For the MR approach, $\bar{E}/\bar{S}$ exhibits a non-monotonic trend due to the concave behavior of $\bar{S}$. Similarly to Fig.~\ref{fig:wur_vs_benchmarks}, there exists a value of $P$ (or equivalently $Q$) that maximizes $\bar{p}_q$ and $\bar{p}_a$. The minimum $\bar{E}/\bar{S}$ of approximately $0.068$~mJ is achieved at $P=30$, where $\bar{p}_a \approx 0.88$ and $\bar{p}_q \approx 0.83$, corresponding to the most energy-efficient operating point for MR.
Conversely, for the \gls{wur} approach, $\bar{E}/\bar{S}$ decreases monotonically with $P$,  indicating that energy reduction dominates, reaching up to a 42\% decrease in energy consumption per served packet at $P=35$. This behavior makes the \gls{wur} solution particularly attractive in scenarios where a high priority is placed on push traffic, such as anomaly reporting~\cite{chiariotti2026combined}.
Finally, although the benchmark schemes can accommodate a larger number of queries within a frame—e.g., for $P=5$, $Q'=34$ under MR and $Q''=40$ for RR and \gls{fsa}, compared to $Q=7$ for the \gls{wur} approach due to the $k_w$ slots required for \gls{wus} reception—this advantage is achieved at the expense of a higher energy cost per successfully delivered packet.}

\section{Conclusions}
\label{sec:conclusion}
In this work, we introduced a novel dual-mode communication framework \revise{enabling per-device} dynamic adaptation to integrate data traffic in push- and pull-based communication. 
We proposed a unified time-frame \gls{mac} structure in which time is divided into sub-frames reserved for push- and pull-based communication. An extensive performance analysis to demonstrate the fundamental trade-offs in \revise{push/pull performance} has been provided \revise{quantitatively demonstrating how prioritizing push-based traffic ensures the collection of informative data, while reducing the system's capacity to respond to external queries, and vice versa.} Finally, \revise{\textcolor{black}{by allowing the} direct implementation of} ultra-low power \glspl{wur} at the nodes, our proposed approach enables \revise{a reduction of up to 42\% in energy consumption per served packet} compared with the traditional communication approaches. 
We believe our introduced framework can serve as a foundation for application-driven, agile \gls{mac} protocol designs for next-generation wireless networks.
\revise{To that end, future work will explore learning-based methods to estimate traffic load and design optimal dynamic resource allocation policies that meet performance goals, accounting for multi-frame deadlines and heterogeneous node capabilities.}

\appendices
\section{Proof of Lemma~\ref{lemma:alarm-success-given-attempts}}
\label{sec:proof-alarm-success}
Let $s$ and $c$ be the \glspl{rv} representing the number of push slots with \emph{exactly one alarm packet} trying to access and the number of slots \emph{with more than one alarm}, respectively. This means that $s = n_s(t)$, bwhile $c$ denotes the number of push slots where collisions happen. The frame structure defined in Sec.~\ref{sec:model:frame} constrains $s + c \le P$. The joint probability distribution of $s$ and $c$ results in~\cite{israel_math}:
\begin{equation}
\label{eq:P(s,c|np)}
\begin{split}
    \mathcal{P}(s, c | n_p(t)) &= \Bigl( \frac{P - s + 1 - c}{P}\Bigr) \mathcal{P}(s-1, c | n_p(t)-1) + \\
    &+ \frac{c}{P} \mathcal{P}(s, c | n_p(t)-1) + \frac{s+1}{P} \mathcal{P}(s+1, c-1 | n_p(t)-1),
\end{split}
\end{equation}
given the initial condition $\mathcal{P} (0,0|0) = 1$~\cite{israel_math}.
Eq.~\eqref{eq:P(s,c|np)} is obtained assuming $n_p(t) - 1$ alarmed nodes have already selected their slot and a new alarmed node has to select it. For this, three disjoint events exist~\cite{israel_math}: 
$1$) $s-1$ push slots are selected by exactly one alarmed node while $c$ are selected by multiple ones; then, the new alarmed node must select any of the $P - (s - 1) - c$ free push slots; $2$) $s$ push slots are selected by exactly one alarmed node while $c$ are selected by multiple ones; then, the new alarm node must select one of the $c$ push slots; $3$) $s+1$ slots are selected by exactly one alarmed node while $c - 1$ are selected by multiple ones; then \textcolor{black}{a new alarm must select one of the $s + 1$ push slots.}

We compute the marginal \gls{pmf} related to the successful transmission of alarm packets, given $n_p(t)$, summing over $c$ in eq.~\eqref{eq:P(s,c|np)} and obtaining~\eqref{eq: pmf_ns_nf_np}, which \textcolor{black}{completes} the proof. \hfill \qed

\section{Proof of Lemma~\ref{lemma:pmf_ny}}
\label{sec:proof-pmf_ny}
Having defined $n_y(t)$ as the number of non-alarmed nodes not receiving any \gls{wus} and $n_z(t)$ as the number of non-alarmed nodes receiving a \gls{wus}, we can write the following relations connecting these \glspl{rv} to the total number of alarmed nodes, $n_a(t)$, the number of nodes receiving a \gls{wus} regardless of their state, $n_u(t) = \min(Q, n_q(t))$, and the number of alarmed nodes receiving a \gls{wus}, $n_w(t)$:
\begin{equation}
\begin{cases}
    n_u(t) = \min(Q, n_q(t)) = n_z(t) + n_w(t), \\
    N = n_z(t) + n_y(t) + n_a(t).
    \label{eq:ny-nz-relations}
\end{cases}    
\end{equation}

Using~\eqref{eq:ny-nz-relations}, the conditional \gls{pmf} of $n_z(t)$ is
\begin{equation}
    \begin{aligned}
    &\mathcal{P}(n_z(t) = z| n_w(t) = j, n_q(t) = q) = 
    \begin{cases}
        1, &\text{if } z = \min(Q, q) - j, \\
        0, &\text{otherwise}.
    \end{cases}
    \end{aligned}
    \label{eq:prob_z_unalarmed_WuS}
\end{equation}
which does not depend on $n_a(t)$. Similarly, the conditional \gls{pmf} of $n_y(t)$ results in
\begin{equation}
\begin{split}
&\mathcal{P} (n_y(t) = y| n_z(t) = z, n_a(t) = i) =
\begin{cases} 
    1, &\text{if } y = N - z - i,   \\    
    0, &\text{otherwise},
\end{cases}
    \end{split}
    \label{eq: prob(n_y=y|nq, na)}
\end{equation}
which does not depend on $n_w(t)$ and $n_q(t)$. Using the law of total probability to relate the \gls{pmf} of $n_y(t)$ to its conditional \gls{pmf}, we obtain eq.~\eqref{eq: prob_ny} and complete the proof. \qed

\footnotesize
\bibliographystyle{IEEEtran}
\bibliography{bibl}

@ARTICLE{mini-slot,
  author={Zhang, Jiawei and others},
  journal={J. Lightw. Technol.}, 
  title={Low Latency {DWBA} Scheme for Mini-Slot Based {5G} new Radio in a Fixed and Mobile Converged {TWDM-PON}}, 
  year={2022},
  volume={40},
  number={1},
  pages={3-13},
  doi={10.1109/JLT.2021.3117972}}

@techreport{3gpp:38-912,
 author = {3GPP},
 day = {},
 institution = {{3rd Generation Partnership Project (3GPP)}},
 month = {07},
 note = {v14.1.0},
 number = {38.912},
 title = {{Study on New Radio (NR) access technology}},
 type = {{Tech. rep. (TR)}},
 year = {2017},
 
}

@techreport{3gpp:38-869,
 author = {3GPP},
 day = {},
 institution = {{3rd Generation Partnership Project (3GPP)}},
 month = {10},
 note = {Release 18},
 number = {38.869},
 title = {{Study on low-power Wake-up Signal and Receiver for NR}},
 type = {{Tech. rep. (TR)}},
 year = {2023},
}

@ARTICLE{grandient-opt,
  author={Pu, Shi and others},
  journal={IEEE Trans. Autom. Control}, 
  title={Push–Pull Gradient Methods for Distributed Optimization in Networks}, 
  year={2021},
  volume={66},
  number={1},
  pages={1-16},
  doi={10.1109/TAC.2020.2972824}}

@INPROCEEDINGS{talli2024push,
  author={Talli, Pietro and Mason, Federico and Chiariotti, Federico and Zanella, Andrea},
  booktitle={IEEE INFOCOM 2024 - IEEE Conf. Comput. Commun. Workshops (INFOCOM WKSHPS)}, 
  title={Push- and Pull-based Effective Communication in Cyber-Physical Systems}, 
  year={2024},
  volume={},
  number={},
  pages={1-7},
  doi={10.1109/INFOCOMWKSHPS61880.2024.10620815}}

@inproceedings{wagner2023low,
  title={Low-Power Wake-Up Signal Design in {3GPP} Release 18},
  author={Wagner, Sebastian and others},
  booktitle={2023 IEEE Conf. Standards Commun. Netw. (CSCN)},
  pages={222--227},
  year={2023},
  organization={IEEE}
}

@ARTICLE{israel_math,
  author={Leyva-Mayorga, Israel and Tello-Oquendo, Luis and Pla, Vicent and Martinez-Bauset, Jorge and Casares-Giner, Vicente},
  journal={IEEE Trans. Wireless Commun.}, 
  title={On the Accurate Performance Evaluation of the {LTE-A} Random Access Procedure and the Access Class Barring Scheme}, 
  year={2017},
  volume={16},
  number={12},
  pages={7785-7799},
  doi={10.1109/TWC.2017.2753784}}

@INPROCEEDINGS{cavallero2024co-existence,
  author={Cavallero, Sara and others},
  booktitle={2024 IEEE 25th Int. Workshop Signal Process. Adv. Wireless Commun. (SPAWC)}, 
  title={Coexistence of Pull and Push Communication in Wireless Access for {IoT} Devices}, 
  year={2024},
  volume={},
  number={},
  pages={841-845},
  doi={10.1109/SPAWC60668.2024.10694005}}

@article{hoglund20243gpp,
  title={{3GPP} release 18 wake-up receiver: Feature overview and evaluations},
  author={H{\"o}glund, Andreas and others},
  journal={IEEE Commun. Standards Mag.},
  volume={8},
  number={3},
  pages={10--16},
  year={2024},
  publisher={IEEE}
}

@inproceedings{chiariotti2024distributed,  
    author={Chiariotti, Federico and Munari, Andrea and Badia, Leonardo and Popovski, Petar},   
    booktitle={IEEE INFOCOM 2025},    
    title={Distributed Optimization of Age of Incorrect Information with Dynamic Epistemic Logic},    
    year={2025},   
    volume={},   
    number={},   
    pages={1-10},  
doi={10.1109/INFOCOM55648.2025.11044748}}

@article{tamura2019low_all,
  title={Low-overhead wake-up control for wireless sensor networks employing wake-up receivers},
  author={Tamura, Naoki and Yomo, Hiroyuki},
  journal={IEICE Trans. Commun.},
  volume={102},
  number={4},
  pages={732--740},
  year={2019},
}

@article{agheli2025integratedpushandpullupdatemodel,
  author={Agheli, Pouya and Pappas, Nikolaos and Popovski, Petar and Kountouris, Marios},
  journal={IEEE Trans. Commun.}, 
  title={Integrated Push-and-Pull Update Model for Goal-Oriented Effective Communication}, 
  year={2025},
  volume={73},
  number={11},
  pages={10914-10928},
  doi={10.1109/TCOMM.2025.3587033}}

@article{akar2024query,
   author={Akar, Nail and Ulukus, Sennur},
  journal={IEEE Trans. Commun.}, 
  title={Query-Based Sampling of Heterogeneous {CTMCs}: Modeling and Optimization With Binary Freshness}, 
  year={2024},
  volume={72},
  number={12},
  pages={7705-7714},
}

@article{shiraishi2022query,
  title={Query timing analysis for content-based wake-up realizing informative {IoT} data collection},
  author={Shiraishi, Junya and others},
  journal={IEEE Wireless Commun. Lett.},
  volume={12},
  number={2},
  pages={327--331},
  year={2022},
  publisher={IEEE}
}

@article{kalor2025data,
  author={Kal{\o}r, Anders E. and Popovski, Petar and Huang, Kaibin},
  journal={IEEE Trans. Commun.}, 
  title={Data Sourcing Random Access Using Semantic Queries for Massive {IoT} Scenarios}, 
  year={2025},
  volume={73},
  number={12},
  pages={14381-14396},
  doi={10.1109/TCOMM.2025.3610171}}

@article{lin2022survey,
  title={A survey on {DRX} mechanism: Device power saving from {LTE} and {5G} new radio to {6G} communication systems},
  author={Lin, Kuang-Hsun and Liu, He-Hsuan and Hu, Kai-Hsin and Huang, An and Wei, Hung-Yu},
  journal={IEEE Commun. Surv. Tut.},
  volume={25},
  number={1},
  pages={156--183},
  year={2022},
  publisher={IEEE}
}

@article{rostami2020wake,
  title={Wake-up radio-based {5G} mobile access: Methods, benefits, and challenges},
  author={Rostami, Soheil and Kela, Petteri and Leppanen, Kari and Valkama, Mikko},
  journal={IEEE Commun. Mag.},
  volume={58},
  number={7},
  pages={14--20},
  year={2020},
  publisher={IEEE}
}

@INPROCEEDINGS{Shiraishi_2024_globecom,
  author={Shiraishi, Junya and Cavallero, Sara and Pandey, Shashi Raj and Saggese, Fabio and Popovski, Petar},
  booktitle={IEEE Global Commun. Conf. (GLOBECOM)}, 
  title={Coexistence of Push Wireless Access with Pull Communication for Content-based Wake-up Radios}, 
  year={2024},
  volume={},
  number={},
  pages={4836-4841},
  doi={10.1109/GLOBECOM52923.2024.10901717}}

@article{pandey2025mediumaccesspushpulldata,
  author={Pandey, Shashi Raj and Saggese, Fabio and Shiraishi, Junya and Chiariotti, Federico and Popovski, Petar},
  journal={IEEE Commun. Mag.}, 
  title={Medium Access for Push-Pull Data Transmission in {6G} Wireless Systems}, 
   year={2026},
  volume={64},
  number={4},
  pages={172-178},
  doi={10.1109/MCOM.001.2500233}}

@ARTICLE{Aloha_push,
  author={Wang, Jiwen and Yu, Jihong and Chen, Xiaoming and Chen, Lin and Qiu, Changquan and An, Jianping},
  journal={IEEE Trans. Commun.}, 
  title={Age of Information for Frame Slotted Aloha}, 
  year={2023},
  volume={71},
  number={4},
  pages={2121-2135},
  doi={10.1109/TCOMM.2023.3244214}}

@article{piyare2017ultra,
  title={Ultra low power wake-up radios: A hardware and networking survey},
  author={Piyare, Rajeev and others},
  journal={IEEE Commun. Surv. Tut.},
  volume={19},
  number={4},
  pages={2117--2157},
  year={2017},
  publisher={IEEE},
}

@ARTICLE{wieselthier1989aloha,
  author={Wieselthier, J.E. and Ephremides, A. and Michaels, L.A.},
  journal={IEEE Trans. Commun.}, 
  title={An exact analysis and performance evaluation of framed {ALOHA} with capture}, 
  year={1989},
  volume={37},
  number={2},
  pages={125-137},
  doi={10.1109/26.20080}}

@article{ruiz2022energy,
  title={Energy-efficient wake-up signalling for machine-type devices based on traffic-aware long short-term memory prediction},
  author={Ru{\'\i}z-Guirola, David E and Rodr{\'\i}guez-L{\'o}pez, Carlos A and Montejo-S{\'a}nchez, Samuel and Souza, Richard Demo and L{\'o}pez, Onel LA and Alves, Hirley},
  journal={IEEE Internet Things J.},
  volume={9},
  number={21},
  pages={21620--21631},
  year={2022},
  publisher={IEEE}
}

@ARTICLE{Aloha_mmtc,
  author={Moradian, Masoumeh and Dadlani, Aresh and Khonsari, Ahmad and Tabassum, Hina},
  journal={IEEE Trans. Commun.}, 
  title={Age-Aware Dynamic Frame Slotted ALOHA for Machine-Type Communications}, 
  year={2024},
  volume={72},
  number={5},
  pages={2639-2654},
  doi={10.1109/TCOMM.2024.3351362}}

@article{froytlog2019ultra,
  title={Ultra-low power wake-up radio for {5G} {IoT}},
  author={Froytlog, Anders and others},
  journal={IEEE Commun. Mag.},
  volume={57},
  number={3},
  pages={111--117},
  year={2019},
  publisher={IEEE}
}

@ARTICLE{rostami2019wake_model,
  author={Rostami, Soheil and Lagen, Sandra and Costa, Mário and Valkama, Mikko and Dini, Paolo},
  journal={IEEE Transactions on Communications}, 
  title={Wake-Up Radio Based Access in {5G} Under Delay Constraints: Modeling and Optimization}, 
  year={2020},
  volume={68},
  number={2},
  pages={1044-1057},
  doi={10.1109/TCOMM.2019.2954389}}

@inproceedings{agheli2024effective,
  title={Effective communication: When to pull updates?},
  author={Agheli, Pouya and Pappas, Nikolaos and Popovski, Petar and Kountouris, Marios},
  booktitle={ICC 2024-IEEE Int. Conf. Commun.},
  pages={183--188},
  year={2024},
  organization={IEEE}
}

@ARTICLE{OsamaEnergyLatency,
  author={Elgarhy, Osama and Reggiani, Luca and Alam, Muhammad Mahtab and Zoha, Ahmed and Ahmad, Rizwan and Kuusik, Alar},
  journal={IEEE Access}, 
  title={Energy Efficiency and Latency Optimization for {IoT} {URLLC} and {mMTC} Use Cases}, 
  year={2024},
  volume={12},
  number={},
  pages={23132-23148},
  doi={10.1109/ACCESS.2024.3364349}}

@article{Metzger2019Modeling,
title={Modeling of Aggregated {IoT} Traffic and Its Application to an {IoT} Cloud},
author={Florian Metzger and T. Hossfeld and André Bauer and Samuel Kounev and P. Heegaard},
journal={Proc. IEEE},
year={2019},
volume={107},
pages={679-694},
doi={10.1109/jproc.2019.2901578}}

@article{RUIZGUIROLA2025126726,
title = {Modeling {IoT} traffic patterns: Insights from a statistical analysis of an {MTC} dataset},
journal = {Expert Systems with Applications},
volume = {272},
pages = {126726},
year = {2025},
issn = {0957-4174},
doi = {https://doi.org/10.1016/j.eswa.2025.126726},
author = {David E. Ruiz-Guirola and Onel L.A. López and Samuel Montejo-Sánchez},
}

@techreport{3gpp_poisson,
 author = {3GPP},
 day = {},
 institution = {{3rd Generation Partnership Project (3GPP)}},
 month = {10},
 note = {},
 number = {37.868},
 title = {{RAN Improvements for Machine-type Communications}},
 type = {{Technical report (TR)}},
 year = {2011},
 url = {}
}

@INPROCEEDINGS{cavallero_urllc,
  author={Cuozzo, Giampaolo and Cavallero, Sara and Pase, Francesco and Giordani, Marco and Eichinger, Joseph and Buratti, Chiara and Verdone, Roberto and Zorzi, Michele},
  booktitle={EuCNC/6G Summit}, 
  title={Enabling {URLLC} in {5G} {NR} {IIoT} Networks: A Full-Stack End-to-End Analysis}, 
  year={2022},
  volume={},
  number={},
  pages={333-338},
  doi={10.1109/EuCNC/6GSummit54941.2022.9815750}}

@article{zakeri2025goal,
author = {Zakeri, Abolfazl and Moltafet, Mohammad and Codreanu, Marian},
year = {2025},
month = {01},
pages = {1-1},
title = {Goal-Oriented Remote Tracking Through Correlated Observations in Pull-based Communications},
volume = {PP},
journal = {IEEE Commun. Lett.},
doi = {10.1109/LCOMM.2025.3613934}
}

@ARTICLE{vilni2025real,
  author={Sadeghi Vilni, Saeid and Wichman, Risto},
  journal={IEEE Commun. Lett.}, 
  title={Real-Time Tracking System With Partially Coupled Sources}, 
  year={2025},
  volume={29},
  number={9},
  pages={2158-2162},
  doi={10.1109/LCOMM.2025.3588270}}

@inproceedings{chiariotti2026combined,
  title={A Combined Push-Pull Access Framework for Digital Twin Alignment and Anomaly Reporting},
  author={Chiariotti, Federico and Saggese, Fabio and Munari, Andrea and Badia, Leonardo and Popovski, Petar},
  booktitle={IEEE INFOCOM ASoI 2026-IEEE Conf. Comput. Commun.},
  pages={1--7},
  year={2026},
  organization={IEEE}
}

@ARTICLE{Agheli2026pullquery,
  author={Agheli, Pouya and Pappas, Nikolaos and Kountouris, Marios},
  journal={IEEE Trans. Commun.}, 
  title={Pull-Based Query Scheduling for Goal-Oriented Semantic Communication}, 
  year={2026},
  volume={74},
  number={},
  pages={3845-3857},
  doi={10.1109/TCOMM.2026.3657444}}

\end{document}